\documentclass[letterpaper, 11pt]{article}
\usepackage[margin=1in]{geometry}

\def\cocoon{0}

\ifnum\cocoon=0
\usepackage{amsthm}  
\fi
\usepackage{amsmath}
\usepackage{amssymb}
\usepackage{xcolor}

\usepackage[noend]{algorithmic}
\usepackage[linesnumbered,ruled,vlined]{algorithm2e}
\usepackage{mathtools}
\usepackage{amsopn}
\usepackage{proba} 
\usepackage{amsfonts}
\usepackage{nicefrac}
\usepackage{float}
\usepackage{enumitem}
\usepackage{commath}
\usepackage{thmtools}
\usepackage{thm-restate}
\usepackage[capitalize]{cleveref}
\usepackage{csquotes}

\SetKwComment{Comment}{$\triangleright$\ }{}
\SetKwFor{Loop}{Loop}{}{}

\DeclareMathOperator{\lcs}{\textsf{LCS}}
\DeclareMathOperator{\lis}{\textsf{LIS}}

\DeclareMathOperator{\disj}{\textsf{DISJ}}
\DeclareMathOperator{\disjnk}{\textsf{DISJ}_{n,k}}

\newcommand{\Cone}{C^{(1)}}
\newcommand{\Ctwo}{C^{(2)}}

\newcommand{\path}{\mathcal{P}}

\ifnum\cocoon=0

\newtheorem{lemma}{Lemma}[section]
\newtheorem{theorem}{Theorem}
\newtheorem{corollary}{Corollary}[section]
\newtheorem{definition}{Definition}[section]
\newtheorem{conjecture}{Conjecture}[section]

\newtheorem{claim}{Claim}[section]

\newtheorem{remark}[theorem]{Remark}
\fi

\newcommand{\zo}{\ensuremath{{\{0,1\}}}\xspace}

\newcommand{\eps}{\ensuremath{\varepsilon}}

\def\*#1{\mathbf{#1}}
\def\+#1{\mathcal{#1}}

\renewcommand{\algorithmiccomment}[1]{\bgroup\hfill $\triangleright$~#1\egroup}
\newcounter{ALC@tempcntr}
\def\final{0} 
\ifnum\final=0  
\newcommand{\authnote}[3]{\textcolor{#2}{{\sf (#1's Note: {\sl{#3}})}}}

\else 
\newcommand{\note}[1]{}
\fi

\usepackage[capitalize]{cleveref}

\usepackage{tikz}
\usetikzlibrary{shapes}
\usetikzlibrary{positioning}
\usetikzlibrary{fit}
\usetikzlibrary{calc}
\usetikzlibrary{backgrounds}
\usetikzlibrary{patterns}
\usetikzlibrary{matrix}
\usetikzlibrary{arrows}

\usepackage{adjustbox}

\usepackage[noadjust]{cite}

\newcommand{\ynote}{\authnote{Yu}{green}}

\newcommand{\tx}{\tilde{x}}

\begin{document}
\title{Streaming and Query Once Space Complexity of Longest Increasing
Subsequence}
\author{
 \and Xin Li\thanks{Department of Computer Science,  Johns Hopkins University, Email: lixints@cs.jhu.edu. Supported by NSF CAREER Award CCF-1845349 and NSF Award CCF-2127575.}
 \and Yu Zheng\thanks{Meta Platforms Inc., Email: hizzy1027@gmail.com  Partially supported by NSF CAREER Award CCF-1845349. Most of the work was done while the author was a graduate student at Johns Hopkins University.}}
\date{}

\maketitle
\begin{abstract}
Longest Increasing Subsequence ($\lis$) is a fundamental problem in combinatorics and computer science. Previously, there have been numerous works on both upper bounds and lower bounds of the time complexity of computing and approximating $\lis$, yet only a few on the equally important space complexity.

In this paper, we further study the space complexity of computing and approximating $\lis$ in various models. Specifically, we prove non-trivial space lower bounds in the following two models: (1) the adaptive query-once model or read-once branching programs, and (2) the streaming model where the order of streaming is different from the natural order. 

As far as we know, there are no previous works on the space complexity of $\lis$ in these models. Besides the bounds, our work also leaves many intriguing open problems.   

\ifnum\cocoon=1
\fi

\end{abstract}

\section{Introduction}
Longest Increasing Subsequence ($\lis$) is a natural measure of a sequence where the alphabet has a total order, and finding (the length of) $\lis$ in a sequence is a fundamental problem in both combinatorics and computer science, which has been studied for decades. For example, the well known Erd\"{o}s–Szekeres theorem in combinatorics states that, given any natural numbers $r$ and $s$, any sequence of at least $(r-1)(s-1)+1$ objects with a total order contains a monotonically increasing subsequence of length $r$ or a monotonically decreasing subsequence of length $s$. Besides being interesting in its own right, $\lis$ is also closely related to other important string problems. For example, it is a special case of the problem of finding the longest common subsequence ($\lcs$) between two strings, where one string is arranged in the increasing order. As such, algorithms for $\lis$ are often used as subroutines for $\lcs$, which in turn has wide applications in bio-informatics due to its connections to gene sequences.

In terms of computing $\lis$, the classical Patience Sorting algorithm \cite{Hammersley72, doi:10.1137/1004036} can find an $\lis$ of a sequence of length $n$ over an alphabet $\Sigma$ in time $O(n \log n)$ and space $O(n \log n)$, while the work of \cite{DBLP:journals/mst/KiyomiOOST20} generalizes this by providing a trade-off between the time and space used. Specifically, for any $s \in \N$ with $\sqrt{n} \leq s \leq n$, \cite{DBLP:journals/mst/KiyomiOOST20} gives an algorithm that uses $O(s \log n)$ space and $O(\frac{n^2}{s} \log n)$ time for computing $\lis$-length, and $O(\frac{n^2}{s} \log^2 n)$ time for finding an actual subsequence. However, no algorithm is known to achieve a better trade-off. We remark that the the decision version of $\lis$ is in the class $\mathsf{NL}$ (non-deterministic logspace), and thus by Savitch's theorem \cite{savitch1970relationships}, it can be solved in time $n^{O(\log n)}$ and space $O(\log^2 n)$. 

Several works studied the problem of \emph{approximating} $\lis$, with the goal of achieving either better time complexity or better space complexity. For time complexity, one aims to obtain a good approximation of $\lis$ using \emph{sublinear time}. The known results in this category depend heavily on the length of $\lis$. For example, the work of \cite{doi:10.1137/130942152} provides an $(1+\eps)$ approximation of $\lis$-length in truly
sublinear time if the length is at least $(1-\lambda)n$ for $\lambda = \Omega(\log \log n/ \log n)$. When the length of $\lis$ is at least $\lambda n$ for an arbitrary $\lambda < 1$, \cite{DBLP:conf/focs/RubinsteinSSS19} gave an algorithm that provides an $O(1/\lambda^3)$ approximation in $\Tilde{O}(\sqrt{n}/\lambda^7)$ time. A subsequent work of \cite{DBLP:conf/soda/MitzenmacherS21} improved the approximation to $O(1/\lambda^{\eps})$ and the time complexity to $O(n^{1-\Omega(\eps)}(\log n/\lambda)^{O(1/\eps)})$. The work of \cite{newman2021new} introduced an $O(1/\lambda)$ approximation algorithm with non-adaptive query complexity $\Tilde{O}(\sqrt{r}\text{ poly}(1/\lambda))$ assuming there are $r$ distinct values in the sequence. Finally, a recent work of \cite{9996890} achieves a $1/\lambda^{o(1)}$ approximation in $O(n^{o(1)}/\lambda)$ time for any $\lambda=o(1)$. We note that all these algorithms are randomized algorithms.

The situation becomes better for space complexity.\ Here, the goal is to obtain a good approximation of $\lis$ using \emph{sublinear space}, while still maintaining polynomial time. In this context, the work of \cite{saha17} provides an algorithm for $\lis$-length that achieves an $\eps n$ additive approximation using space $O(\frac{\log n}{\eps})$, while a recent work \cite{cheng_et_al:LIPIcs.ICALP.2021.54} achieves a $1+\eps$ approximation using $\tilde{O}_{\eps,\delta}(n^\delta)$ space and $\tilde{O}_{\eps, \delta}(n^{2-2\delta})$ time, for any constants $\delta\in (0,\frac{1}{2})$ such that $\frac{1}{\delta}$ is an integer, and $\eps\in (0,1)$. \cite{cheng_et_al:LIPIcs.ICALP.2021.54} further provides an algorithm that achieves a $1+O(\frac{1}{\log \log n})$ approximation using $O(\frac{\log^4 n}{\log \log n})$ space and $n^{5+o(1)}$ time. 

In addition, due to applications on large data sets, $\lis$ is also well studied in the \emph{streaming model}, where the sequence is accessed from a data stream with one or a small number of passes, rather than random access. In this model, the work of \cite{DBLP:conf/soda/GopalanJKK07} provides a one-pass streaming algorithm that achieves a $1+\eps$ approximation of $\lis$-length, using $O(n \log n)$ time and $O(\sqrt{n/\eps}\log n)$ space, where the space complexity is known to be tight due to the lower bounds given in \cite{gal2010lower, ergun2008distance}. Interestingly, all these algorithms are deterministic, and it is an intriguing question to see if randomized algorithms can achieve better space complexity.

In this work, we further study the lower bounds of space complexity for computing and approximating $\lis$ in various models. Recall that the decision version of $\lis$ is in the class $\mathsf{NL}$, and the question of whether $\mathsf{NL}=\mathsf{L}$ ($\mathsf{L}$ stands for deterministic logspace) is still a major open problem in theoretical computer science. Hence, to get any non-trivial space lower bound, it is natural and necessary to restrict the models. Here, we study two different models: the query-once model and the streaming model.

\paragraph{Query-once model.} In this model, we allow the algorithm to have random access to the input sequence, but impose the restriction that the algorithm can only query each element of the sequence at most once, using any adaptive strategy. It is thus a natural and strict generalization of the one-pass streaming model. 

In fact, we study the slightly more general model of \emph{read-once branching program}, which can be used to represent the deterministic query-once algorithm. Informally, a read-once branching program models the query-once model as a directed graph, where at each node of the branching program, the program queries one position of the input sequence. Depending on the queried input symbol, the program jumps to another node and continues the process. The read-once property ensures that in any computation path, any input element is queried at most once. The size of the branching program is defined as the number of nodes, which roughly corresponds to $2^{O(s)}$ for a space $s$ computation. A formal description is given in \cref{sec:prelim}. We note that the model of read-once branching program is a non-uniform model, hence is more general than the uniform query-once algorithm model.

Size lower bounds of read-once branching programs for explicit functions have also been the subject of extensive study. Following a long line of research \cite{Wegener88, Zak84, DBLP:conf/fct/Dunne85, DBLP:journals/tcs/Jukna88, DBLP:journals/tcs/KrauseMW91, Simon1992ANL, doi:10.1137/S0097539795290349, DBLP:journals/ipl/Gal97, DBLP:journals/ipl/BolligW98, AndreevBCR99, DBLP:journals/tcs/Kabanets03, Li23}, the current best lower bound for a function in $\mathsf{P}$ is $2^{n-O(\log n)}$ \cite{Li23}, which is optimal up to the constant in $O(\cdot)$. There are even functions in uniform-$\mathsf{AC}^0$ that give strong size lower bounds for read-once branching programs \cite{DBLP:journals/tcs/Jukna88, DBLP:journals/tcs/KrauseMW91, DBLP:journals/ipl/Gal97, DBLP:journals/ipl/BolligW98, LiZ23}, where the current best lower bound is $2^{(1-\delta)n)}$ for any constant $\delta>0$ \cite{LiZ23}.

However, many of the above lower bounds are achieved by somewhat contrived functions, thus it is also important and interesting to study the size of read-once branching programs computing natural functions. Notable examples include the integer multiplication function \cite{doi:10.1137/S0097539795290349, DBLP:conf/stoc/BolligW01, ABLAYEV200378}, where a lower bound of $\Omega(2^{n/4})$ \cite{DBLP:conf/stoc/BolligW01} is known for deterministic read-once branching programs and a lower bound of $2^{\Omega(n/\log n)}$ \cite{ABLAYEV200378} is known for randomized read-once branching programs; and the clique-only function \cite{Zak84, BorodinRS93}, where a lower bound of $2^{\Omega(\sqrt{n})}$ is known even for non-deterministic read-once branching programs.

Following this direction, in this paper we study the size of deterministic read-once branching programs for $\lis$-length.

\paragraph{Streaming model with arbitrary order.} In this model, the input sequence is again given to the algorithm in a data stream. However, unlike the standard streaming model for $\lis$ where the sequence is given from the first element to the last element, here we study the streaming model where the elements are given in some arbitrary order, according to a permutation of $[n]$. Indeed, in practice the data stream containing the input sequence may not be exactly in the natural order of the elements. For example, consider the situation of an asynchronous network, where a client sends a long sequence to a server for processing. Even if the client sends the sequence from the first element to the last element, the elements received by the server may be in a different order due to transmission delays in the network. Hence, this model is also a natural and practical generalization of the standard streaming model for $\lis$. In this paper we study two types of streaming orders, which we define later. The motivation of these orders comes from the fact that a random order falls into either type with high probability. Therefore, these orders actually capture most of the streaming orders.

We study $\lis$ space lower bounds for both deterministic and randomized algorithms in these models. Previously, the only known non-trivial space lower bounds for approximating $\lis$ are in the streaming model with standard order, where the aforementiond works \cite{gal2010lower, ergun2008distance} give a lower bound of $\Omega\left (\frac{1}{R}\sqrt{\frac{n}{\eps}}\log \left ( \frac{|\Sigma|}{\eps n}\right ) \right )$ for any $R$ pass deterministic streaming algorithm achieving a $1+\eps$ approximation of $\lis$-length when the alphabet size $|\Sigma| \geq n$, and the subsequent work of \cite{li_et_al:LIPIcs.FSTTCS.2021.27} which extends this bound to $\Omega(\min(\sqrt{n}, |\Sigma|)/R)$ for any constant $\eps>0$ with any alphabet size. In contrast, there is no known non-trivial space lower bound for randomized algorithms that $1+\eps$ approximate $\lis$-length, for any constant $\eps>0$ in the streaming model.

For exact computation of the $\lis$-length, \cite{sun2007communication} establishes a space lower bound of $\Omega(n)$ for any $O(1)$-pass randomized streaming algorithm, as long as $|\Sigma| \geq n$, again in the standard order. 

In summary, we stress that the only known lower bounds for either approximating or exact computation of the $\lis$-length in the streaming model are in the \emph{standard order}, as far as we know. Furthermore, there is no known space lower bound (even for deterministic computation of $\lis$-length) in the query-once model.
\subsection{Our results}
In this paper we establish new space lower bounds for computing and approximating $\lis$-length in the query-once model and the streaming model. We start by stating our results in the query-once model. Below, for a sequence $x \in \Sigma^n$, we use $\lis(x)$ to stand for the length of a longest increasing subsequence in $x$. In this paper we always assume that the alphabet size $|\Sigma| > n$.

\paragraph{Query-once model}


\begin{restatable}{theorem}{ROBPThm}
\label{ROBPThm}
Given input sequences $x\in \Sigma^n$, any read-once branching program that computes $\lis(x)$ has size $2^{\Omega(n)}$. 
\end{restatable}

\begin{remark}
Since our alphabet size $|\Sigma| > n$, our input size is actually $n'=O(n \log n)$ and hence in terms of $n'$, the lower bound is $2^{\Omega(n'/\log n')}$. Also, in our model the read-once branching program reads a symbol in $\Sigma$ each time, instead of just one bit. It is an interesting open question to see if one can get better lower bounds or in the model where the read-once branching program reads an input bit each time.
\end{remark}

This gives the following corollaries.

\begin{restatable}{corollary}{QueryOnceThm}
\label{QueryOnceThm}
Given input sequences $x\in \Sigma^n$, any deterministic algorithm that computes $\lis(x)$ and queries each symbol of $x$ at most once (the queries can be adaptive) needs to use $\Omega(n)$ space. 
\end{restatable}

\begin{restatable}{corollary}{StreamExtThm}
\label{StreamExtThm}
Given input sequences $x\in \Sigma^n$, any one-pass deterministic streaming algorithm computing $\lis(x)$ needs to use $\Omega(n)$ space, regardless of the order of the elements in the stream. 
\end{restatable}

It is clear that these lower bounds are almost tight, up to a $\log |\Sigma| $ factor.

\begin{remark}
We note that \cite{sun2007communication} establishes a space lower bound of $\Omega(n)$ for any $O(1)$-pass randomized streaming algorithm that computes $\lis(x)$, as long as $|\Sigma| \geq n$.\ However, we stress that their result does NOT supersede ours, since their result only applies to the standard streaming order, where the input sequence is read from left to right. On the other hand, our Corollary~\ref{StreamExtThm} applies to any arbitrary streaming order. Therefore, these two results are incomparable. Furthermore, our Theorem~\ref{ROBPThm} and Corollary~\ref{QueryOnceThm} give lower bounds in the read-once branching program model and query-once model, which are strictly stronger than the streaming model.
\end{remark}

\paragraph{Streaming model in special orders}

Given an alphabet $\Sigma$ and an input sequence $x= x_1x_2 \cdots x_n\in \Sigma^n$, we represent the order of streaming as a permutation $\pi: [n] \to [n]$ and write $\pi = \pi_1\pi_2 \cdots \pi_n$, where each $\pi_i =\pi(i) \in [n]$. The streaming algorithm has access to $x$ in the order of $\pi$, i.e., it sees $x_{\pi_1}$, then $x_{\pi_2}$ and so on. In other words, the index $i$ refers to the $i$-th symbol in the stream, while the index $\pi_i$ refers to the $\pi_i$-th symbol in the original input sequence $x$, so the $i$-th symbol in the stream corresponds to the $\pi_i$-th symbol in the original input sequence $x$.

We prove space lower bounds for two types of orders. 

\begin{restatable}[Type 1 order]{definition}{TypeOneOrder}
 $\pi$ is a type 1 order with parameter $m$ if there are two sets of indices $I = \{i_1, i_2, \dots, i_m\}$ and $J=\{j_1, j_2, \dots, j_m\}$, with $|I|=|J|=m$ such that $\max(I) < \min(J)$ and $1 \le \pi_{i_1} < \pi_{j_1}< \pi_{i_2} <\pi_{j_2} < \cdots <  \pi_{i_m} <\pi_{j_m} < n $.
\end{restatable}

Notice that $\max(I) < \min(J)$ guarantees that for any $i \in I $ and $j\in J$, $x_{\pi_i}$ appears before $x_{\pi_j}$ in order $\pi$. The constraint $1 \le \pi_{i_1} < \pi_{j_1}< \pi_{i_2} <\pi_{j_2} < \cdots <  \pi_{i_m} <\pi_{j_m} < n $ says that the original indices of the symbols corresponding to $I$ and $J$ in $x$ are interleaved. 

For example, any order that first reveals all symbols in odd positions and then all symbols in even positions is an order of Type 1 with parameter $n/2$ since we can let $I = \{1,2, \dots, n/2\}$ and $J=\{n/2 +1, n/2+2 , \dots, n\}$. 

We have the following lower bounds for deterministic approximation algorithms in Type 1 order.

\begin{restatable}{theorem}{TypeOneOrderThm}
\label{TypeOneOrderThm}
Given input sequences $x\in \Sigma^n$ in any \textbf{type 1} order with parameter $m$,  any $R$-pass deterministic algorithm that achieves a $1 + 1/32$ approximation of $\lis(x)$ needs to use $\Omega(m/R)$ space. 
\end{restatable}


We believe that type 1 order is interesting, since one can show that, with high probability, a random streaming order is a type 1 order with parameter $\Omega(n)$. Hence type 1 order actually captures most of the streaming orders. In turn, this gives the following corollary. 

\begin{restatable}{corollary}{RandomStrmOrder}
\label{RandomStrmOrder}
Given input sequences $x\in \Sigma^n$ in a random order sampled uniformly from all permutations on $[n]$, with probability $1-2^{-\Omega(n)}$, any $R$-pass deterministic algorithm that achieves a $1 + 1/32$ approximation of $\lis(x)$ needs to use $\Omega(n/R)$ space.
\end{restatable}

For randomized algorithms, we show a simple lower bound for exact computation.

\begin{restatable}{theorem}{TypeOneOrderThmRand}
Given input sequences $x\in \Sigma^n$ in any \textbf{type 1} order with parameter $m$, any $R$-pass randomized algorithm that computes $\lis(x)$ correctly with probability at least $2/3$ needs to use $\Omega(m/R)$ space. 
\end{restatable}

The second type of orders generalizes type 1 orders, and corresponds to orders with interleaving blocks. 

\begin{restatable}[Type 2 order]{definition}{TypeTwoOrder}
$\pi$ is a type 2 order with parameters $r$ and $s$ if the following holds.  There are $r$ disjoint sets of indices $B_1, B_2 , \dots, B_r$ each of size $s$ such that $$\max\big( \bigcup_{l \text{ is odd}} B_l \big)< \min\big( \bigcup_{l \text{ is even}} B_l \big), $$ 
and for any $1\le l < r$, we have $\max_{i\in B_l}(\pi_i) < \min_{i\in B_{l+1}}(\pi_i)$. 

\end{restatable}

For example, if we divide $[n]$ evenly into $\sqrt{n}$ blocks each of size $n$, then any order that first reveals symbols in the odd blocks and then symbols in the even blocks is a Type 2 order with $r = s = \sqrt{n}$. This is because we can pick $r=s=\sqrt{n}$ such that $B_1, B_2, \dots, B_{r/2}$ are the odd blocks and $B_{r/2+1}, B_{r/2+2}, \dots, B_{r}$ are the even blocks. 

\paragraph{Remark} Type 1 order with parameter $m$ can be viewed as a special case of type 2 order as it is essentially type 2 order with parameter $r = 2m$ and $s = 1$.

We have the following lower bound for deterministic approximation algorithms in type 2 order. 

\begin{restatable}{theorem}{TypeTwoOrderThm}
\label{thm:lis_type2}
Given input sequences $x\in \Sigma^n$ in any \textbf{type 2} order with parameters $r$ and $s$, any $R$-pass deterministic algorithm that gives a $1+1/400$ approximation of $\lis(x)$ needs to use $\Omega(r\cdot s / R)$ space.
\end{restatable}

Intuitively, the streaming order that is most friendly to computing or approximating $\lis$-length is the natural order (or the reverse order). Indeed, in these models the algorithm given in \cite{DBLP:conf/soda/GopalanJKK07} achieves a one-pass $1+\eps$ approximation of $\lis$-length, using $O(n \log n)$ time and $O(\sqrt{n/\eps}\log n)$ space. We thus conjecture that this is the best one can do, and $1+\eps$ approximation of $\lis$-length in any streaming order requires $\Omega(\sqrt{n})$ space. Specifically, we have the following conjecture, which seems quite natural but we haven't been able to prove.

\begin{conjecture}
Given input sequences $x\in \Sigma^n$, for any one-pass streaming order, any deterministic algorithm that achieves $1+\eps$ approximation of $\lis(x)$ needs to use $\Omega_{\eps}(\sqrt{n})$ space.
\end{conjecture}

In particular, it is not clear if there is any streaming order where one can get a constant factor deterministic approximation algorithm for $\lis$-length that uses $o(\sqrt{n})$ space.

\subsection{Technique overview}

\ifnum\cocoon=1
We now give an overview of the techniques used in our paper. Full proofs are deferred to the appendix.

\paragraph{Query-once model and read-once branching programs.}
\fi

\ifnum\cocoon=0
\subsubsection{Query-once model and read-once branching programs.}
\fi

Our first lower bound is for any read-once branching program. Specifically, we prove that any read-once branching program computing $\lis$ exactly must have a large size, which in turn implies a space lower bound. Assume the alphabet size is $m$ and the input length is $n$, and let $x$ and $y$ be two different input strings. Let $I_j(x)$ denote the set of positions the branching program queries at the $j$-th level on input $x$. Note that $I_j(x)$ may not be equal to $I_j(y)$, due to the adaptivity of the algorithm. 

However, we show that if the computation paths of two different inputs $x$ and $y$ go through the same node at the $j$-th level, then we must have $I_j(x) = I_j(y)$. This is because if not, we can find two new sequences $x'$ and $y'$ such that $x'$ (resp. $y'$) follows the same computation path as $x$ (resp. $y$) until the $j$-th level, and at the same time $x'$ and $y'$ follow the same computation path after the $j$-th level. Since the branching program can query each position of an input at most once, the computation path of $x'$ and $y'$ after the $j$-th level can not query any position in $I_j(x) \cup I_j(y)$. That means the branching program will output the same result for $x'$ and $y'$, and at the same time, there must be at least one position in $x'$ (and $y'$) not queried by the branching program. We can now change the symbol in that unqueried position of $x'$ to get another sequence with a different $\lis$-length, but since this position is not queried the branching program will still give the same output. This is a contradiction. 

With the above observation, we show that there must be a level in the branching program with $2^{\Omega(n)}$ nodes. The proof is based on the following two claims. 

First, for any two different increasing sequences $x, y \in \Sigma^n$ and a subset $S$ of $[n]$ with size $n/5$, if $x$ and $y$ are not equal when restricted to $S$ (i.e. $x|_S\neq y|_S$), then we can find two new sequences $x'$ and $y'$ with unequal $\lis$-length such that for positions in $S$,  $x' = x$ and $y' = y$ (i.e. $x'|_S = x|_S$, $y'|_S = y'|_S$) and for positions not in $S$, $x' = y'$ (i.e. $x'|_{[n] \setminus S} = y'|_{[n] \setminus S}$) (Claim~\ref{clm:query-once1}
\ifnum\cocoon=1
in the appendix\fi). The construction of $x'$ and $y'$ is given in the proof of Claim~\ref{clm:query-once1}. 

Second, there exists a set of $2^{\Omega(n)}$ increasing sequences such that for any $S \subseteq [n]$ with size $n/5$ and any two sequences in the set, they are not equal when restricted to $S$ (Claim~\ref{clm:query-once2} \ifnum\cocoon=1
in the appendix\fi). The proof is based on a probabilistic argument, where we show that by independently randomly choosing $2^{\Omega(n)}$ increasing sequences, there is a non-zero probability that they satisfy the claim.

Now, consider the set in Claim~\ref{clm:query-once2} and the $n/5$-th level of the branching program, we argue that any two sequences in the set cannot go through the same node at that level. This is because if they do, then by Claim~\ref{clm:query-once1}, we can build two new sequences with different $\lis$-length but the branching program will give the same output. Thus, the $n/5$-th level must have $2^{\Omega(n)}$ nodes. This yields our $2^{\Omega(n)}$ size lower bound for read-once branching programs and $\Omega(n)$ space lower bound any adaptive query-once algorithm. 

\ifnum\cocoon=1
\paragraph{Streaming in special orders.}
\fi
\ifnum\cocoon=0
\subsubsection{Streaming in special orders.}
\fi

As many other streaming space lower bounds, our proof is based on reductions from communication complexity problems. Specifically, we consider a 2-party communication problem where Alice and Bob each holds a different part of the input sequence based on the streaming model. In addition, our proof uses error-correcting codes to create gaps that are necessary for our approximation lower bounds. A more detailed description is given below.

\textbf{Type 1 order.} In type 1 order with parameter $m$, there are $2m$ positions and the streaming first reveals the odd positions and then the even positions. To make the presentation easier, let's consider the special case where we first see all the odd positions of the input sequence, and then all the even positions. We can translate this streaming order into a 2-party communication problem, where Alice holds all the odd positions of the input sequence and Bob holds all the even positions. Their goal is to approximate the $\lis$-length. The space required by any streaming algorithm is at least the communication complexity between Alice and Bob.

The communication complexity lower bound is established by constructing a large fooling set. Specifically, our fooling set is obtained from a simple binary asymptotically good error-correcting code $C$ with codeword length $n/4$. The size of the code $C$ (number of codewords) is $2^{\Omega(n)}$ and for any two different codewords $c_1$ and $c_2$, their Hamming distance is $\Omega(n)$. 

Assume Alice holds a codeword $u \in \zo^{n/4}$ and Bob holds a codeword $v \in \zo^{n/4}$. Our proof gives a construction that transforms $u$ and $v$ into a sequence such that the odd positions only depend on $u$ and the even positions only depend on $v$. Denote this new sequence by $z = z(u,v)$. We divide $z$ into $n/4$ small blocks of size 4, such that, for the $i$-th block, if $u_i = v_i$, or $u_i = 0$ and $v_i = 1$, then the $\lis$-length of this block is 1; on the other hand, if $u_i  = 1 $ and $v_i = 0$, then the $\lis$-length of this block is 0. We design the blocks so that the total $\lis$-length of $z$ is the summation of the $\lis$-lengths of all blocks. Thus, when $u = v$, the $\lis$ of $z$ is always a fixed number. But when $u \neq v$, since $u$ and $v$ are both codewords of $C$, there is a constant fraction of positions such that $u_i\neq v_i$. Then the $\lis$-length of one of $z(u,v) $ and $z(v,u)$ is smaller by a constant factor. This makes $C$ a fooling set and gives the $\Omega(n)$ space lower bound. Generalizing this to any parameter $m < n$, we get the lower bound of $\Omega(m)$ for type 1 order with parameter $m$.

We also show that any randomized algorithm computing $\lis$-length exactly in this order must use $\Omega(m)$ space.\ The proof is based on a reduction from the set-disjointness problem to computing $\lis$ exactly.

\textbf{Type 2 order.} In type 2 order, we assume Alice and Bob each holds $r$ interleaved blocks of size $s$. Since type 1 order is a special case of type 2 order with $r=m$ and $s=1$, naturally, we want to extend our previous techniques to type 2 order and construct another fooling set.


Recall that for type 1 order, our fooling set is obtained from an asymptotically good error-correcting code, where Alice and bob each holds a codeword, $u$ and $v$. The input sequence constructed from these codewords is divided into $m$ small blocks, where the $i$-th block is determined by $u_i$ and $v_i$. Whenever $u_i \neq v_i$, we can potentially reduce the $\lis$-length by 1. Since the code $C$ has $\Omega(m)$ distance, this gives a lower bound for constant factor approximation algorithms. 

In the case of type 2 order, each block has size $s$. If we can only create a gap of 1 in $\lis$-length for each pair of blocks depending on whether the corresponding bits of $u$ and $v$ are equal, then we can only get a lower bound for $1+o(1)$ approximation. To amplify the gap, we use another asymptotically good error-correcting code.  

Here we present a simplified version of our construction to illustrate the high-level idea, while the actual construction is slightly more complicated. We use two asymptotically good error-correcting codes, $\Cone$ and $\Ctwo$. $\Cone\subseteq \{0,1\}^s$ is a binary code, which has codeword length $s$ and distance $s/4$. Note that $\Cone$ is a subset of $\{0,1\}^s$. $\Ctwo\subseteq (\Cone)^r$ uses $\Cone$ as its alphabet, and it has codeword length $r$ and distance $r/2$. Since both $\Cone$ and $\Ctwo$ are asymptotically good, we have $|\Cone| = 2^{\Omega(s)}$ and $|\Ctwo| = 2^{\Omega(r\cdot s)}$.

We assume Alice and Bob each holds a codeword of $\Ctwo$, denoted by $A$ and $B$. The first step is to construct a $s\times 2r$ weighted directed grid graph using $A$ and $B$. The grid has $2r$ columns and each column has $s$ nodes. In this graph, each node has two outgoing edges, one connecting to the node on its right, and the other connecting to the node below. The graph has the following properties. All edges going downward has weight 1, and all edges going from an even column to an odd column has weight 1. For edges going from an odd column to an even column, the weights depend on $A$ and $B$. Specifically, for an edge going from the $j$-th node in the $(2i-1)$-th column to the $j$-th node in the $2i$-th column, it has weight 1 if $(A_i)_j \ne (B_i)_j$ and 0 otherwise. Here, $A_i$ is the $i$-th symbol of $A$, which is a codeword of $\Cone$; and $(A_i)_j$ the $j$-th symbol of $A_i$, which is in $\{0,1\}$. The same notation applies to $(B_i)_j$. 

If $A = B$, then any edge going from an odd column to its right have weight 0. Otherwise, by the distance property of our codes, there are at least $r/2$ odd columns such that for each column, at least 1/4 of the edges going out of this column to its right must have weight 1. This is because the code $\Ctwo$ has distance $r/2$ and $\Cone$ has distance $s/4$.   

We now transform the graph into a 2-party communication problem where Alice holds the odd columns and Bob holds the even columns. Their goal is to approximate the largest weight of any path going from the top-left node to the bottom-right node. We can use a combinatorial argument based on the Erd\"{o}s-Szekeres theorem to show that if $A = B$, the largest weight of any path is a fixed number; however, if $A \ne B$, the largest weight of any path increases by a constant factor.  Thus, the code $\Ctwo$ gives a fooling set for this problem and any deterministic algorithm that solves the problem needs communication complexity at least $\log(|\Ctwo|) = \Omega(r\cdot s)$. 

Finally, we reduce this problem to approximating $\lis$-length in type 2 order, by assigning appropriate values to each node of the grid graph and reordering the nodes into a sequence. The high-level idea is that for each path in the grid graph, we can find a sequence whose $\lis$-length is equal to the weight of the path. This gives our space lower bound in type 2 order.





\subsection{Related work}
There have also been several works studying the “complement” problem of $\lis$, namely approximating the distance to
monotonicity, i.e., $d_m(x) = n - \lis(x)$, in both the sublinear-time and the streaming settings (note that computing $d_m(x)$ exactly is equivalent to computing $\lis(x)$). For time complexity, this was first studied by \cite{Ailon2007EstimatingTD}, and a subsequent work of \cite{doi:10.1137/130942152} gave a $(1+\eps)$ approximation algorithm in time $\mathsf{poly}(1/d_m(x), \log n)$ for any constant $\eps>0$.

In the streaming setting, \cite{DBLP:conf/soda/GopalanJKK07} and \cite{ergun2008distance} gave algorithms that achieve $O(1)$ approximation of $d_m(x)$ using $\mathsf{polylog}(n)$ space. This was later improved by \cite{saks2013space} to achieve  a randomized $(1+\eps)$-approximation algorithm using $\mathsf{polylog}(n)$ space, and further by \cite{doi:10.1137/1.9781611973730.83} to achieve a deterministic $(1+\eps)$-approximation algorithm using $\mathsf{polylog}(n)$ space. \cite{doi:10.1137/1.9781611973730.83} also proved streaming space lower bound of $\Omega(\log^2 n/\eps)$ and $\Omega(\log^2 n/(\eps \log\log n))$ for deterministic and randomized $(1+\eps)$-approximation of $d_m(x)$, respectively. 

$\lis$ has also been recently studied in the Massively Parallel Computation (MPC) model and  the fully dynamic model. In the former, \cite{im2017efficient} gave a $O(1/\eps^2)$-round algorithm that achieves $(1+\eps)$-approximation of $\lis$-length, as long as each machine uses space $n^{3/4+\Omega(1)}$. In the latter, a sequence of works \cite{DBLP:conf/stoc/MitzenmacherS20, DBLP:conf/stoc/GawrychowskiJ21, DBLP:conf/stoc/KociumakaS21} resulted in an exact algorithm \cite{DBLP:conf/stoc/KociumakaS21} with sublinear update time, together with a deterministic $1+o(1)$-algorithm algorithm with update time $n^{o(1)}$.

\section{Open Problems}
Our work leaves several natural open problems. First, can one get better lower bounds for the query-once model/read-once branching programs? Can one generalize the bounds to query-$k$ times/read-$k$ times branching programs? Can one get any lower bounds for randomized or non-deterministic branching programs? How about lower bounds for approximation?

Second, can one get any lower bounds for randomized algorithms approximating $\lis$-length in the streaming model? Finally, is our conjecture in the streaming model true, or is there any streaming order where one can get a constant factor deterministic approximation of $\lis$-length that uses $o(\sqrt{n})$ space?

\ifnum\cocoon=0
\paragraph{Organization of the paper.} The rest of the paper is organized as follows. After some preliminaries in \cref{sec:prelim}, we give our lower bounds in the query-once model in \cref{sec:queryonce}, and our lower bounds in the streaming model in \cref{sec:stream}.
\fi

\section{Preliminaries}

\label{sec:prelim}

\paragraph{Notations.}

We use $[n]$ to denote the set of all positive integers that are at most $n$, and use $[a, b]$ to denote the set of all integers that are at least $a$ and at most $b$. 

For any set of indices $I  \subseteq [n]$, we use $x|_I$ to denote the subsequence of $x$ restricted to the indices in $I$.  In other words, assume $I = \{i_1, i_2, \dots, i_t\}$ and $1\le i_1 < i_2 < \cdots < i_t \le n$, then  $x|_I = x_{i_1}x_{i_2}\cdots x_{i_t}$.

\paragraph{Branching program. }

The following definition is from \cite{beame1989general}. An $R$-way branching program consists of a directed acyclic rooted graph of out-degree $R=R(n)$ with each non-sink node labelled by an index from $\{1,2,\dots, n\}$ and with $R$ out-edges of each node labelled $1,\dots, R $. Edges of the branching program may also be labelled by a sequence of values from some output domain. The \emph{size} of a branching program is the number of nodes it has. 

Let $x=(x_1, x_2, \dots, x_n)$ be an $n$-tuple of integers chosen from the range $[R]$. An $R$-way branching program computes a function of input $x$ as follows. The computation starts at the root of the branching program. At each non-sink node $v$ encountered, the computation follows the out edge labelled with the value of $x_i$ where $i$ is the index that labels node $v$ (i.e. variable $x_i$ is queried at $v$). The computation terminates when it reaches a sink node. The sequence of nodes and edges encountered is the \emph{computation path} followed by $x$. The output of the branching program on input $x$ is determined by the sink node the computation path ends in. 

The \emph{time} used by a branching program is the length of the longest computation path followed by any input. The \emph{space} used by a branching program is the logarithm base 2 of its size. 

An $R$-way branching program is \emph{levelled} if the nodes of the underlying graph are assigned levels so that the root has level $0$ and the out edges of a node at level $l$ only go to nodes at level $l +1$.  It is known that any branching program can be levelled without changing its time and with at most squaring its size \cite{pippenger1979simultaneous}. This will not change the time used and will increase the space complexity by a factor of 2 at most. Thus we can assume $R$-way branching programs are levelled without loss of generality. If the branching program has the additional property that in any computation path, any $x_i$ is queried at most once, then it is called a read-once branching program.

\paragraph{The log sum inequality.}

In \cref{sec:lb_type_2}, we will use the following inequality.

\begin{lemma}[The log sum inequality]
\label{lem:log-sum}
Given $2n$ positive numbers $a_1, a_2, \dots, a_n $ and $b_1, b_2, \dots, b_n$, let $a = \sum_{i=1}^n a_i$ and $b = \sum_{i=1}^n b_i$, then we have
\[ \sum_{i=1}^n a_i\log \frac{a_i}{b_i}\ge a\log \frac{a}{b}.\]

\end{lemma}

\paragraph{Communication Complexity and fooling set.} 

We will consider the $2$-party communication model where $2$ players Alice and Bob each holds input $x\in X$, $y \in Y$ respectively. The goal is to compute a function $f: X\times Y \rightarrow \{0,1\} $. We define the \emph{deterministic communication complexity} of $f$ in this model as the minimum number of bits required to be sent by the players in every deterministic communication protocol that always outputs a correct answer. Correspondingly, the \emph{randomized communication complexity} (denoted by $R_{\epsilon}(f)$) is the minimum number of bits required to be sent by the players in every randomized communication protocol that can output a correct answer with probability at least $1-\epsilon$.

Some of our proofs use the classical fooling set argument. Fooling set is defined as following.

\begin{definition}[Fooling set]
Let $f: X\times Y \rightarrow \{0,1\} $. A subset $S\subseteq X\times Y$ is a fooling set for $f$ if there exists $z\in \{0,1\}$ such that $\forall (x, y) \in S$, $f(x,y) = z$ and for any two distinct $ (x_1,y_1), (x_2, y_2)\in S$, either $f(x_1,y_2) \ne z$ or $f(x_2, y_1) \ne z$. 
\end{definition}

We have the following Lemma.

\begin{lemma} [Corollary 4.7 of \cite{roughgarden2016communication}]
    If $f$ has a fooling set $S$ of size $t$, then the deterministic communication complexity of $f$ is at least $\log_2 t$.
\end{lemma}

\paragraph{Set-disjointness Problem.}

We will use some classical results about the communication complexity of 2 party set-disjointness problem. 2 party set-disjointness problem is defined as following. Assuming there are two parties, Alice and Bob, each of them holds a $k$-subset of $[n]$, i.e. Alice holds a set $A\subseteq [n]$ with $|A| = k$ and Bob holds another set $B\subseteq [n]$ with $|B| = k$. Alice and Bob wants to compute whether $A$ and $B$ are disjoint ($\disjnk$). Here, 

\begin{equation*}
    \disjnk(A,B) = \begin{cases}
    1, \text{ if } A\cap B = \emptyset, \\
    0, \text{ if } A\cap B \ne \emptyset.
    \end{cases}
\end{equation*}

We have the following result about the randomized communication complexity of set-disjointness.

\begin{theorem}
[Theorem 1.2 of \cite{haastad2007randomized}]
\label{thm:cc_lb_set_disj}
For any $c<1/2$, the randomized communication complexity $R_{1/3}(\disjnk) = \Omega(k)$ for every $k\le cn$.
\end{theorem}

\paragraph{Erd\"{o}s-Szekeres theorem. }

\begin{theorem}[Erd\"{o}s-Szekeres theorem]
\label{thm:erdos-szekeres}
    Given two integers $r>0$ and $s > 0$, any sequence of distinct real numbers with length at least $(r-1)(s-1) + 1$ contains a monotonically increasing subsequence of length $r$ or a monotonically decreasing subsequence of length $s$.
\end{theorem}

\section{Computing LIS in the Query-Once Model}\label{sec:queryonce}

In this section, we consider the \emph{query-once model} where algorithms only allowed to access each symbol of the input sequence once. We show that in this model, any algorithm that computes $\lis$ exactly must use $\Omega(n)$ space.  
\ROBPThm*

We introduce some notations. For any set of indices $S\subseteq [n]$, the function $f_S$ takes two sequences $x, y\in \Sigma^n$ as input and outputs a sequence $\sigma$ such that $\sigma|_S = x|_S$ and $\sigma|_{[n]\setminus S} = y|_{[n]\setminus S}$. So if $\sigma = f_S(x,y)$, then for any $i\in [n]$, 

\begin{equation*}
    \sigma_i = \begin{cases}
    x_i, \text{ if } i\in S,\\
    y_i, \text{ if } i\notin S.
    \end{cases}
\end{equation*}

We now show the following 2 claims.

\begin{claim}
\label{clm:query-once1}
Assume the alphabet size $m > c \cdot n$ for some large enough constant $c$, and $x,y\in [m]^n$ are two increasing sequences. Then $\forall \ S\subset [n]$ with $|S| = n/5$ such that $x|_S \neq y|_S$, there is a sequence $z\in [0, m+1]^n$ such that 
\[\lis(f_S(x,z)) \neq \lis (f_S(y,z)).\]
\end{claim}

\begin{proof}[Proof of Claim~\ref{clm:query-once1}]

In other words, the claim says if two increasing sequences $x$ and $y$ are not equal on a subset of indices $S$ with $|S| = n/5$, then there is a way to build two sequences $\Tilde{x},\Tilde{y}\in [0,m+1]^n$ with $\Tilde{x}|_S = x|_S$, $\Tilde{y}|_S = y|_S$ and $\Tilde{x}|_{[n]/S} = \Tilde{y}|_{[n]/S}$ such that $\lis(\Tilde{x})\neq \lis(\Tilde{y})$.

Let us fix index set $S\in [n]$ and increasing sequences $x,y$ such that $x|_S \neq y|_S$. Our goal is two build $\Tilde{x}$ and $\Tilde{y}$ satisfying the above properties. We first fix $\Tilde{x}|_S = x|_S$, $\Tilde{y}|_S = y|_S$ and say that the positions not in $S$ are unfixed. 

Let $l \in S $ be the first index that $x_l\neq y_l $ and $r$ be the last position that $x_r\neq y_r$. There are two cases. First, there are at least $2n/5$ unfixed positions after $l$ (with index larger than $l$). Second, there are less than $2n/5$ unfixed positions after $l$. 

For the first case, without loss of generality, we assume $x_l < y_l$. Let $E$ be the set of symbols that appeared in $x|_S$ or $y|_S$ with value larger than $x_l$, i.e. 
\[E = \{a \; | \; a > x_l  \text{ and } a = x_i \text{ or } a = y_i \text{ for some } i \in S\}\]

Let $u_E$ be the sequence of all elements in $E$ concatenated in the increasing order. Since $|S| = n/5$, we have $|E| = |u_E| \le 2n/5-1$. We can build $\Tilde{x}$ and $\Tilde{y}$ as follows. First, we put $0$ into all unfixed positions before $l$ (if any). Then, for unfixed positions after $l$, we first put $u_E$ into the unfixed positions (notice that we assume there are at least $2n/5$ unfixed positions after $l$ and $|u_E|\le 2n/5-1$) and for the remaining unfixed positions, we put $m+1$ into them. 

We argue that $\lis(\Tilde{x})\neq \lis(\Tilde{y})$. To see this, notice that $\Tilde{x}$ and $\Tilde{y}$ are equal before $l$-th position and all symbols before $l$-th position are smaller than $x_l$. Also notice that, in both $\Tilde{x}|_{[l+1:n]}$ and $\Tilde{y}|_{[l+1:n]}$, there are exactly $|E|+1$ distinct symbols (symbols in $E$ plus the symbol $m+1$) and we can find an increasing susbequence of length $|E|+1$ ($u_E$ plus the symbol $m+1$). Since all symbols after $l$-th position are larger than $x_l$ (in both $\Tilde{x}$ and $\Tilde{y}$). For $\Tilde{x}$, the longest increasing subsequence is the longest increasing subsequence before $l$-th position plus the symbol $x_l$, and plus the longest increasing subsequence after $l$-th position. For $\Tilde{y}$, the longest increasing subsequence we can find is similar except we can not include the symbol $x_l$ (which does not appear in $\Tilde{y}$). Thus, we have  $\lis(\Tilde{x}) =  \lis(\Tilde{y}) + 1$.

For the second case, the argument is symmetric. There are $4n/5$ unfixed positions. If there are less than $2n/5$ unfixed positions after $l$. There are at least $2n/5$ unfixed positions before $r$. Without loss of generality, we assume $x_l > y_l$ and let $E$ be the set of symbols that appeared in $x|_S$ or $y|_S$ with value smaller than $x_r$. Similarly, let $u_E$ be the sequence of all elements in $E$ concatenated in the increasing order.

We can build $\Tilde{x}$ and $\Tilde{y}$ as follows. First, we put $m+1$ into all unfixed positions after $r$ (if any). Then, for unfixed positions before $r$, we first put a $0$ in the first unfixed position and then put  $u_E$ into the unfixed positions. If there are any remaining unfixed positions, we put $0$ into them. By a similar analysis in the first case, we can show that $\lis(\Tilde{x}) =  \lis(\Tilde{y}) + 1$. This finishes the proof.

\end{proof}

\begin{claim}
\label{clm:query-once2}
Assume the alphabet size $m > c \cdot n$ for some large enough constant $c$, then there exists a set $T$ of increasing subsequences in $[m]^n$ with size $2^{\Omega(n)}$, such that $\forall \ S\subset [n]$ with $|S| = n/5$ and $\forall \ x, y\in T$, we have $x|_S \neq y|_S$.
\end{claim}

\begin{proof}[Proof of Claim~\ref{clm:query-once2}]  

We prove the claim by a probabilistic argument. Let $A$ be the set of all increasing sequence in $[m]^n$ ($|A| = \binom{m}{n}$) and $x, y$ are two sequences sampled i.i.d. from $A$ uniformly. Then, consider a fixed set $S\subset [n]$ with $|S| = n/5$.  We now give an estimation of the probability of $x|_S = y|_S$. Before the computation, we introduce some notations. For simplicity,let $t = n/5$. Let $S = \{i_1, i_2, \dots, i_{t}\}$ such that $i_1 < i_2 < \cdots < i_t$. In addition, we let $i_0 = 0$ and $i_{t+1} = n+1$. For $i\in [t+1]$, let 
$$d_l= i_l-i_{l-1} -1.$$ 
And for any fixed sequence $z\in A$, we let $z_{t+1} = m+1$ and 
$$a_l(z) = z_{i_l}-z_{i_{l-1}} -1. $$
Clearly, $a_l(z) \ge d_l$ since $z$ is an increasing sequence. Also, we have $\sum_{l=1}^{t+1}d_l = n-t, $ and 
$\sum_{l=1}^{t+1}a_l(z) = m-t $.

We can write $\Pr[x|_S = y|_S] =  \frac{1}{|A|} \sum_{z\in A} \Pr[x|_S = z|_S]$. For a fixed $z$, $\Pr[x|_S = z|_S]$ is exactly the number of sequences in $A$ satisfies $x|_S = z|_S$ divided by the size of $A$. Notice that $x|_S = z|_S$ fixed $t$ positions in $x$. For any $l \in [t]$, there are $\binom{a_l(z)}{d_l}$ ways to pick symbols between $i_l$-th position and $i_{l-1}$-th position. Here, we allow  $i_l = i_{l-1} + 1$ since in this case $\binom{a_l(z)}{d_l} = 1$. For simplicity, we write $a_l= a_l(z) $ for simplicity. Thus, for fixed $z\in A$, 

\begin{align}
\label{eq1}
    \begin{split}
         \Pr[x|_S = z|_S] = & \frac{1}{|A|} \cdot \prod_{l=1}^{t+1}\binom{a_l}{d_l} \\
         \le & \frac{1}{|A|} \cdot \prod_{l=1}^{t+1}\big(\frac{e\cdot a_l}{d_l}\big)^{d_l} \\
         \le &\frac{1}{|A|} \cdot e^{n-t} \cdot \prod_{l=1}^{t+1}\big(\frac{ a_l}{d_l}\big)^{d_l} 
    \end{split}
\end{align}

By the log sum inequality (\cref{lem:log-sum}) and the fact that $\sum_{l=1}^{t+1}d_l = n-t, $ and 
$\sum_{l=1}^{t+1}a_l(z) = m-t $ , we have
\begin{align}
\label{eq2}
\begin{split}
    \log \big(\prod_{l=1}^{t+1}\big(\frac{ a_l}{d_l}\big)^{d_l}\big) &= \sum_{l=1}^{t+1}d_l\log \frac{a_l}{d_l} \\
    &= -\sum_{l=1}^{t+1}d_l\log \frac{d_l}{a_l} \\
    &\le -(n-t)\cdot \log \frac{n-t}{m-t} \\
    & = (n-t)\cdot \log \frac{m-t}{n-t} 
\end{split}
\end{align}

Combine \cref{eq1}, \cref{eq2} and $|A| = \binom{m}{n}$, we have 
\begin{align}
    \begin{split}
        \Pr[x|_S = z|_S] \le &\frac{1}{|A|} \cdot e^{n-t} \cdot \prod_{l=1}^{t+1}\big(\frac{ a_l}{d_l}\big)^{d_l} \\
        \le & \frac{1}{|A|} \cdot e^{n-t} \cdot (\frac{m-t}{n-t})^{n-t} \\
        \le & \big(\frac{n}{m}\big)^n \cdot \Big(e \cdot \frac{m-t}{n-t}\Big)^{n-t} 
    \end{split}
\end{align}

Plug in $t= n/5$, we have 
$
         \Pr[x|_S = z|_S]  \le  \big(\frac{n}{m}\big)^n \cdot \Big(e \cdot \frac{m-t}{n-t}\Big)^{n-t} \le  \big(\frac{n}{m}\big)^{n/5} \cdot \big( \frac{5e}{4}\big)^{4n/5}
$.Thus, for any two $x,y$ sampled i.i.d. from $A$, we have 

\[\Pr[x|_S = y|_S]  \le \big(\frac{n}{m}\big)^{n/5} \cdot \big( \frac{5e}{4}\big)^{4n/5}.\]


Let $I = \{S\subset[n] \text{ such that } |S| = n/5 \}$. We have 

\begin{align}
    \begin{split}
        \Pr[\exists \ S\in I, x|_S = y|_S] \le &   \sum_{S\in I} \Pr[x|_S = y|_S] \\
        \le & \binom{n}{n/5}\cdot \big(\frac{n}{m}\big)^{n/5} \cdot \big( \frac{5e}{4}\big)^{4n/5} \\
        \le & \Big( 5e\cdot \frac{n}{m}\cdot (\frac{5e}{4})^4\Big)^{n/5}.
    \end{split}
\end{align}

Taking $m = \Omega(n)$ to be large enough, we have $\Pr[\exists \ S\in I, x|_S = y|_S] \le 2^{-n/5}$.

If we take a random subset $T\subseteq A$ with size $2^{n/20}$. By a union bound, 

\begin{equation}
    \begin{split}
        \Pr[\forall \ x, y\in T \text{ and }  S\in I, x|_S\neq y|_S] \ge & 1 - \binom{|T|}{2}\cdot 2^{-n/5} > 0.
    \end{split}
\end{equation}

Thus, there exists a set $T\subset [m]^n$ of increasing sequences with size $2^{n/20}$ such that for any two sequences $x,y\in T$ and any $S\in I$, we have $x|_S\neq y|_S$. This proves the claim. 

\end{proof}

\begin{proof} [Proof of \cref{ROBPThm}]
Assume there is a levelled branching program that computes $\lis$ exactly. Let $t = n/5$, we show that the number of nodes at the $t$-th level is $2^{\Omega(n)}$. 

We consider the computation path of different input sequences. Given an input sequence $x$, we let $I_j(x)$ denote the set of positions that have been queried in the first $j$ levels by the branching program. Meanwhile, we say a sequence is even if all its symbols are even numbers. We have the following observation. 

Consider two sequences $x$ and $y$ that takes only even values, if their computation path both go through the same node in $j$-th level, then we must have $I_j(x) = I_j(y)$. To see this, assume $I_j(x) \neq I_j(y)$, let $I_j = I_j(x) \cup I_j(y)$. If $j=n$, $I_j(x) \neq I_j(y)$ means some positions in $x$ or $y$ are not queried. Say there is a position $l$ that is not queried in $x$, then fix any longest increasing subsequence of $x$, we can change $x_l$ to an odd value and increase the $\lis$ by 1. This will not influence the computation path since $x_l$ is not queried. If $1\le j < n$, we can build two new sequences $\Tilde{x}$ and $\Tilde{y}$ by letting $\Tilde{x}|_{I_j} = x|_{I_j}$, $\Tilde{y}|_{I_j} = y|_{I_j}$, and fill any positions in $[n] \setminus I_j$ with symbol 1. Then after $j$-th level, $\Tilde{x}$ and $\Tilde{y}$ will follow the same computation path since after $j$-th level, the branching program will not query any position in $I_j$ (or it will query some position in $x$ or $y$ twice). We must have $\lis(\Tilde{x}) = \lis(\Tilde{y})$. Since the $\lis$ is only determined by positions in $I_j$, we have $x|_{I_j} = y|_{I_j}$. By the same argument as the case of $j = n$, we can get a contradiction. 

Also notice that, for two even increasing sequences $x$ and $y$, if their computation path both go through the same node in $t = n/5$-th level, let $S = I_t(x)$, then we must have $x|_S = y|_S$. This follows from Claim~\ref{clm:query-once1}. To see this, if $x|_S \neq y|_S$, then according to Claim~\ref{clm:query-once1}, there exists $\Tilde{x}$ and $\Tilde{y}$ such that $\Tilde{x}|_S = x|_S$, $\Tilde{y}|_S = y|_S$, $\Tilde{x}|_{[n]/S} = \Tilde{y}_{[n]/S}$ and $\lis(\Tilde{x}) \neq \lis(\Tilde{y})$. If we use $\Tilde{x}$  as input to the branching program, in the first $t$ levels, it will follow the exact same computation path as $x$ since $\Tilde{x}|_S = x|_S$. $\Tilde{y}$ will also follow the same computation path as $y$. Thus, the computation path of $\Tilde{x}$ and $\Tilde{y}$ will collide at $t$-th level. Notice that $\Tilde{x}|_{[n]/S} = \Tilde{y}|_{[n]/S}$, after $t$-th level, $\Tilde{x}$ and $\Tilde{y}$ will follow the same computation path, which means our branching program will output the same result for $\Tilde{x}$ and $\Tilde{y}$.  However, this is contradictory to $\lis(\Tilde{x}) \neq \lis(\Tilde{y})$.

Let $T$ be a set of increasing sequences guaranteed by Claim~\ref{clm:query-once2}. We can turn it into a set of even increasing subsequences by multiplying each symbols by a factor of two. This will not affect the properties guaranteed by Claim~\ref{clm:query-once2} and will only increase the alphabet size by a factor of 2, which is still $O(n)$. We denote this new set by $T'$.  By the above observation, the computation paths of any two sequences in $T'$ must go through different nodes in $t$-th level. Thus, the $t$-th level must have $2^{\Omega(n)}$ nodes. The space used by the branching program is $\Omega(n)$.

\end{proof}

\begin{corollary}
Assume the input sequence $x$ is given as a stream, then no matter what is the order of the stream, any 1-pass deterministic algorithm computing $\lis(x)$ takes $\Omega(n)$ space. 
\end{corollary}

\begin{proof}
    The streaming model can be viewed as a restricted version of the query-once model where the algorithm can only access input sequence in a specific order. Thus, the space lower bound for query-once model also holds for the streaming model. 
\end{proof}

\section{Lower Bounds for Streaming LIS in Different Orders}\label{sec:stream}

In this section, we consider the problem of computing/approximating $\lis$ in the streaming model but with different orders. 

We restate the definition of the orders we studied here.

\TypeOneOrder*

\TypeTwoOrder*

\subsection{Lower Bounds for Type 1 Orders}

\subsubsection{Lower Bounds for Deterministic Algorithms}

\TypeOneOrderThm*

\begin{proof}

Without loss of generality, we assume $n$ is an even number. Consider the order $\pi = 1,3,\dots,n-1,2,4, \dots n $, i.e. in order $\pi$, we first see symbols in odd positions in natural order and then symbols in even positions in natural order. We first show an $\Omega(n)$ lower bound for this order to illustrate our idea. \cref{TypeOneOrderThm} then follows from the proof by a simple observation.

The idea is to build a large fooling set. Let $C\subseteq \{0,1\}^{n/4}$ be an asymptotically good error-correcting code with constant rate and distance $n/16$. We consider the following two party communication scenario: Alice holds a codeword $u \in C$ and Bob holds a codeword $v \in C$.



Consider the following construction. Alice first turns $u\in \{0,1\}^{n/4}$ into $u'\in \N^{n/2}$ by the following transformation. Let $u'$ be an empty sequence at first. Then for each $i\in [n/4]$, let $u_i\in {0,1}$ be the $i$-th symbol of $u$. If $u_i= 0$, we attach $0, 2i-1$ to the end of $u'$ and if $u_i= 1$, we attach $2i-1, 0$ to the end of $u'$. Bob does a similar transformation to get $v' \in \N^{n/2}$. The only difference is that if $v_i=0$, we attach $0,2i$ and if $v_i=1$, we attach $2i, 0$ to the end of $v'$. We note both $u'$ and $v'$ has exactly $n/4$ symbols that are 0 and $n/4$ that are nonzero. 

We can get a sequence $z \in \N^n$ such that 

$$z=u'_1, v'_1, u'_2, v'_2, \dots, u'_{n/2}, v'_{n/2}$$.

In other words, $z$ is equal to $u'$ if we only look at the odd positions of $z$ and is equal to $v'$ if we only look at the even positions. Notice that the sequence $z$ is different if Alice holds $v$ and Bob holds $u$. Thus, we use $z_{u,v}$ to denote the sequence we get when Alice hols $u$ and Bobs holds $v$ and  $z_{v,u}$ vice versa.

We now show that if $u=v$, $\lis( z_{u,v})\ge n/2$. Otherwise, 
$$\min\{\lis\big(z_{u,v}\big), \lis\big(z_{v,u} \big) \} \le \frac{15n}{32}  + 1 .$$ 

To see this, we can divide $z=z_{u,v}$ into  $n/4$ blocks each with 4 symbols. Let $z^i$ be the $i$-th block. By our construction of $z$, $z^i = u'_{2i-1},v'_{2i-1},u'_{2i}, v'_{2i}$. There are four cases:

\begin{itemize}
    \item Case 1: $u_i = v_i = 0$, $z^i = 0,0,2i-1, 2i$.
    \item Case 2: $u_i =1,  v_i = 0$, $z^i = 2i-1,0,0, 2i$.
    \item Case 3: $u_i=0,   v_i = 1$, $z^i = 0,2i,2i-1, 0$.
    \item Case 4: $u_i = v_i = 1$, $z^i = 2i-1, 2i,0,0$.
\end{itemize}

We note that in all cases except case 3,  the two non-zero symbols $2i-1, 2i$ form an increasing sequence. Thus if $u = v$, then each block $z^i$ is either in case 1 or case 4. Thus, $1, 2, 3, \dots, n/2-1, n/2$ is a subsequence of $z$ since $2i-1, 2i$ is a subsequence of block $z^i$. We have $\lis(z)\ge n/2$.

If $u \ne v$, since $u, v$ are both codewords from $C$, by our assumption, the Hamming distance between $u$ and $v$ is at least $n/16$.  Thus, there are at least $n/16$ positions $i\in [n/4]$ such that $z^i$ is in case 2 or case 3. Let $a$ be the number of blocks $z^i$ that are of case 2 and $b$ be the number of blocks $z^i$ that are in case 3. Thus $a+b > n/16$ and one of $a$ or $b$ is at least $n/32$. Without loss of generality, we can assume that $b \ge n/32$ since if not, we can look at the sequence $z_{v,u}$. This is because if $z^i_{u,v}$ is in case 2, then $z^i_{v,u}$ is in case 3 and vice versa. 

Notice that for any $1\le i< j\le n/4$, the nonzero symbols in $z^i$ is always strictly smaller than the nonzero symbols in $z^j$. The best strategy to pick a longest increasing subsequence is to pick the longest increasing subsequence with nonzero symbols in each block $z^i$ and then combine them (with an additional 0 symbol from the first block in some cases). In each block $z^i$, the length of longest increasing subsequence with nonzero symbols is 2 if the block is in case 1, 2, and 4. And it is 1 if the block is in cases 3. By our assumptions that the number of blocks in cases 3 is at least $n/32$. We know $\lis(z) \le \frac{15n}{32} + 1$. 

We can define a function $f:C\times C \rightarrow \{0,1\}$ such that 
 
\begin{equation*}
    f(u,v) = 
    \begin{cases}
    1, \text{ if } \lis(z_{u,v}) \ge n/2\\
    0, \text{ if } \lis(z_{u,v}) \le \frac{15n}{32}  + 1
    \end{cases}
\end{equation*}

We have shown that $C$ is a fooling set for function $f$. The deterministic communication  complexity of $f$ is $\log |C| = \Omega(n)$. Also notice that if we can get a $1+1/32$ approximation of $\lis$ in the order that first reveals odd symbols and then even symbols, we can use it as a protocol to compute $f$ exactly. Thus, any $R$ pass streaming algorithm for this order that gives a $1+1/32$ approximation must use a $\Omega(n/R)$ space. 

Finally, we generalize this result to any order of type 1 with parameter $m$. This is because we can use an error correcting code with codeword length $m/2$ as a fooling set. The above construction gives us a sequence $z\in \N^{2m}$. We can then obtain another sequence $z'$ such that when $z'$ is restricted to indices in $I\cup J$, it is equal to $z$ and all other symbols are fixed to 0.  In the two party communication scenario, Alice holds symbols with indices in $I$ and Bob holds symbols with indices in $J$. Their goal is to compute $\lis(z')$. By the same fooling set argument, we get the lower bound.

\end{proof}

We now show that, with high probability, deterministically approximating $\lis$ of a sequence given in a random streaming order requires $\Omega(n)$ space.

\RandomStrmOrder*

Before the proof of \cref{RandomStrmOrder}, we first show the following claim.

\begin{claim}
\label{clm:random_order_is_type1}
Assume $n$ is a multiple of $32$, and a streaming order is uniformly randomly sampled from all permutations of $[n]$, with probability at least $1-2^{-\Omega(n)}$, it is a type 1 order with parameter $n/32$.
\end{claim}

\begin{proof}
    Let $\pi$ be a streaming order of $n$ elements, which is also a permutation of $[n]$. Denote the set of all permutations over $n$ by $S_n$. Let $A_{\pi}$ be the first $n/2$ elements revealed by the order $\pi$. 

    First, for a subset $A\subseteq [n]$, we consider $A$ can be divided into how many intervals of consecutive integers. In other words, each subset $A$ can be represented as the union of multiple intervals of consecutive integers. Let $g(A)$ be the minimal number of intervals that $A$ can be represented as, i.e.

    \[
    g(A) = \min \{k \ | \ A = \bigcup_{t = 1}^k [i_t,j_t) \text{ for some } i_1< j_1 < i_2 < j_2 <\cdots < i_{k} < j_{k} \}
    \]

    Here, $[i,j)$ is the set of all integers at least $i$ and smaller then $j$. For example, if $A = {1,2,\dots, n/2}$, then $g(A) = 1$. If $A$ is the subset of all even integers in $[n]$, then $g(A) = n/2$.

    For any permutation $\pi$, if $g(A_{\pi}) = m+1$ distinct gaps, then $\pi$ must a type 1 order with parameter $m$. This is because we can represent $A_{\pi}$ as $ \bigcup_{t = 1}^{m+1} [i_t,j_t)$ for some $i_1< j_1 < i_2 < j_2 <\cdots < i_{m+1} < j_{m+1} $. By the definition, we have $j_i \notin A_{\pi}$. We can pick $I = \{i_1, i_2, \dots, i_m\}$ and $J = \{j_1, j_2, \dots, j_m\}$ ($j_{m+1}$ may be equal to $n+1$). Since $I\subseteq A_{\pi}$ and $J\bigcap A_{\pi} = \emptyset$, Any elements in $I$ will be revealed before $J$ in the order of $\pi$. Thus, $\pi$ is a type 1 order with parameter $m$.

    Then, we show that if $A$ is uniformly sampled from the set of all subsets with size $n/2$, then $g(A) \ge n/2$ with probability at least $1-2{-\Omega(n)}$. 

    We compute how many subsets $A$ has $g(A) \le n/32$. Given an integer $k$, consider the how many set $A$ has $g(A) = k$. 
    $A$ can be divided into $k$ intervals. There are $\binom{n/2 - 1}{k-1}$ different ways to divide $n/2$ consecutive integers into $k$ consecutive intervals.

    For elements not in $A$, depending on whether $i_0 = 0$ or $j_k = n+1$, there are $k-1$, $k$ or $k+1$ consecutive intervals. For simplicity, we add 2 dummy elements $0$ and $n+1$, and consider how many ways we can divide $\{0, 1, \dots, n/2 + 1\}$ into $k+1$ consecutive intervals. Similarly to elements in $A$, there are $\binom{n/2 + 1}{k}$ choice of the size of intervals. 

    Thus, the number of subsets $A$ with $g(A) = k$ and $|A| = n/2$ is $\binom{n/2 - 1}{k-1} \cdot \binom{n/2 }{k}$. We can bound the number of sets $A$ with $g(A)\le n/32$ as following

    \begin{align}
    \label{eqn:A_size_bound1}
    \begin{split}
        \Big | \{A \subset [n] \ \big| \ |A| = n/2 \text{ and } g(A) \le n/32\}\Big| &  =\sum_{i = 1}^{n/32} \binom{n/2 - 1}{k-1}\cdot \binom{n/2}{k} \\
        & = \sum_{i = 1}^{n/32} \frac{(k-1)!(n/2-k)!}{(n/2-1)!}\cdot \frac{k!(n/2-k)!}{(n/2!)} \\
        & \le \frac{n}{2} \cdot \sum_{i = 1}^{n/32}\binom{n/2}{k} ^2 \\
        & \le \frac{n}{2} \cdot \Bigg(\sum_{i = 1}^{n/32}\binom{n/2}{k}\Bigg)^2. 
    \end{split}
    \end{align}

    Notice that 

    \begin{equation}
    \label{eqn:binom_bound}
    \sum_{i = 1}^{n/32}\binom{n/2}{k} 2^{-n/2} \le \binom{n/2}{n/32}2^{-15n/32}.
    \end{equation}

    This is because we can view $\sum_{i = 1}^{n/32}\binom{n/2}{k} 2^{-n/2}$ as the probability of doing $n/2$ independent coin flip and there are at most $n/32$ coin flips with head up. This happens if and only if there is a set of $(n/2 - n/32) = 15n/32$ coin flips ends with tail up. There are $\binom{n/2}{15n/32} = \binom{n/2}{n/32}$ choices if subsets with size $15n/32$, and for each, the probability of all coin flips is tail up is $2^{-15n/32}$. Taking a union bound, we get \cref{eqn:binom_bound}. Combining with \cref{eqn:A_size_bound1}, we have the following bound

    \begin{align*}
        \Big | \{A \subset [n] \big| |A| = n/2 \text{ and } g(A) \le n/32\}\Big| & \le \frac{n}{2} \cdot \Bigg(\sum_{i = 1}^{n/32}\binom{n/2}{k}\Bigg)^2 \\
        & \le \frac{n}{2} \cdot \Bigg(\binom{n/2}{n/32} \cdot 2^{n/32}\Bigg)^2 \\
        & \le \frac{n}{2} \cdot \Big(\big(16e\big)^{n/32} \cdot 2^{n/32} \Big)^2  \\
        & \le \frac{n}{2} \cdot \big(32e\big)^{n/16} \\
        & \le \frac{n}{2} \cdot 2^{7n/16} 
    \end{align*}

    There are $\binom{n}{n/2}$ subsets of $[n]$ with size exactly $n/2$. When $n$ is a large enough integer, if we sample a subset $A$ from all subsets of $[n]$ with size $n/2$, we have

    \begin{align*}
        \Pr[g(A) \le n/32]  & = \frac{\Big | \{A \subset [n] \big| |A| = n/2 \text{ and } g(A) \le n/32\}\Big| }{ \binom{n}{n/2} } \\
        & \le \frac{\frac{n}{2} \cdot 2^{7n/16}}{2^{n/2} }\\
        & \le 2^{-n/32}.
    \end{align*}

    Finally, we notice that if $\pi$ is uniformly sampled from $S_n$, then $A_{\pi}$ is uniformly distributed over all subset of $[n]$ with size $n/2$. It is because for any two subsets $A, B\subset [n]$, both with size $n/2$, we have

    \[
    \big | \{\pi\in S_n | A_{\pi} = A\} \big |  = \big | \{\pi\in S_n | B_{\pi} = B\} \big | =\bigg(\big (\frac{n}{2}\big )!\bigg )^2.
    \]
    
    Thus, with probability at lease $1-2^{-n/32}$, the permutation $\pi$ has $g(A_\pi) > n/32$ and is a type 1 order with parameter $n/32$.
    
\end{proof}

\begin{proof}[Proof of \cref{RandomStrmOrder}]

Since with probability at least $1-2^{-\Omega(n)}$, a random permutation (or streaming order) is a type 1 order with parameter $n/32$. By \cref{TypeOneOrderThm},  with probability at least $1-2^{-\Omega(n)}$, a $1+1/32$ deterministic approximation of $\lis(x)$ needs $\Omega(n)$ space.
    
\end{proof}

\subsubsection{Lower Bounds for Randomized Algorithms}

For randomized algorithms, we show an $\Omega(m)$ lower bound for exact computation when the input is given in \textbf{type 1} order with parameter $m$. We prove the following.

\TypeOneOrderThmRand*

\begin{proof}

The idea is to reduce set-disjointness problem to computing $\lis$ exactly.

Consider a \textbf{type 1} order with parameter $m$ and let $I$ and $J$ be the subsets in the definition. We only focus on the subsequence $\tx$ limited to the indices in $I\cup J$, denoted by $\tx = \tx_1 \tx_2 \cdots \tx_{2m}$. The streaming order, when restricted to $\tx$, is $\tx_1,\tx_3,\cdots,\tx_{2m-1},\tx_2,\tx_4,\cdots,\tx_{2m}$. In the 2 party communication model, Alice holds $\tx_1,\tx_3,\cdots,\tx_{2m-1}$ and Bob holds $\tx_2,\tx_4,\cdots,\tx_{2m}$.

Consider an instance $A,B$ of $\disj_{2m,k}$, where $A$ and $B$ are both $k$-subset of $[2m]$. For simplicity, represent $A$ and $B$ as characteristic vectors in $\{0,1\}^m$, i.e. $A$ and $B$ are both 0,1 vector with exactly $k$ 1s.

We construct the sequence $\tx$ such that $\tx_1,\tx_3,\cdots,\tx_{2m-1}$ is determined by $A$ and $\tx_2,\tx_4,\cdots,\tx_{2m}$ is determined by $B$. The construction is straightforward. For every $i\in[m]$, let $\tx_{2i-1} = A_i\cdot 2i$ and $\tx_{2i} = B_i\cdot (2i-1)$.

For simplicity, we attach a dummy $0$ character in front of $\tx$. This won't affect our asymptotic lower bound.  

If $\disj_{2m,k}(A,B) = 1$, for any $i\in[m]$, $A_i$ and $B_i$ must not both be 1. There are $2k$ positions that one of $\tx_{2i-1}$ and $\tx_{2i}$ is non-zero. Since these non-zero positions are increasing, $\lis(\tx) = 2k + 1$ (plus the 0 at the first position).

If $\disj_{2m,k}(A,B) = 0$, there exists $i\in [m]$ such that $A_i = B_i = 1$. We have $\tx_{2i-1} = 2i$ and $\tx{2i} = 2i-1$. Since there are only $2k$ non-zero positions, we have $\lis(\tx) \le 2k$ (plus the 0 at the first position).

Thus, if Alice and Bob can compute $\lis(\tx)$, they can also determine $\disj_{2m,k}(A,B)$. The communication complexity lower bound given in \cref{thm:cc_lb_set_disj} yields our space lower bound.

\end{proof}

\subsection{Lower Bounds for Type 2 Orders}
\label{sec:lb_type_2}

Our communication complexity based argument relies on the analysis of a carefully constructed matrix. In the following, we represent any position in a matrix with a pair of indices $(i,j)$ where $i$ is the row number and $j$ is the column number. For any two distinct positions $(i_1, j_1)$ and $(i_2, j_2)$, we say $(i_1, j_1)< (i_2, j_2)$ if $i_1 < i_2$ \textbf{and} $j_1 < j_2$.

We first show the following lemma, which is a result of Erd\"{o}s-Szekeres theorem.  

\begin{lemma}
\label{lem:erdos-szekeres}
Given a matrix $M\in \{0,1\}^{s\times r}$ with $s$ rows and $r$ columns, if each column of $M$ has exactly $\frac{s}{4}$ 1's, then there exists $t = \lfloor \frac{r\cdot s}{8(r+s)} \rfloor$ positions (denoted by$(i_1, j_1), (i_2, j_2), \dots, (i_t, j_t)$) in the matrix with value 1, such that  $(i_1, j_1) < (i_2, j_2)< \cdots < (i_t, j_t)$.

\end{lemma}

\begin{proof}
    Consider the function $f:[s]\times[r]\rightarrow[rs]$ such that $f(i,j) = r\cdot i - j +1$. $f$ transforms any position $(i,j)$ into an integer in $[rs]$ and the integers corresponding to different positions are distinct. 
    
    For each column, there are exactly $s/4$ positions that have value 1. We transform them into integers in $[rs]$ with function $f$ and arrange them in a decreasing order to get a sequence of integers. We denote the sequence corresponding to $j$-th column by $\sigma^{(j)}$. For example, assume $(i_1, 1), \dots, (i_{s/4}, 1)$ are the positions with value 1 in the first column with $i_1  < \cdots < i_{s/4}$, then $\sigma^{(1)} = r\cdot i_{s/4}  ,  \dots, r\cdot i_1$. 
    
    Let $\sigma = \sigma^{(1)}\circ \cdots \circ \sigma^{(r)}$. So $\sigma$ is the concatenation from $\sigma^{(1)}$ to $\sigma^{(r)}$. We show that there is an increasing subsequence of $\sigma$ with length $t = \lfloor \frac{r\cdot s}{8(r+s)} \rfloor$.
    
    To see this, we first show that the longest decreasing subsequence of $\sigma $ has length at most $r+s$.  Consider any decreasing subsequence of $\sigma$ and the positions corresponding to this decreasing subsequence in the matrix. Assume the first element of this decreasing subsequence is in the range $[r\cdot i+1, r\cdot (i+1)]$ for some integer $i$, for the next element, there are two cases. First, the next element is in the range $[r\cdot i+1, r\cdot (i+1)]$, then it means the element is in the same row but in a column to the right of the first element. Second, the next element is in the same column as the first element, then it must be strictly smaller than $r\cdot i + 1$, which means it is in a row below the first element. This is true for any two consecutive elements in the decreasing subsequence. Notice that $\sigma$ is concatenated by columns from left to right and there are $r$ columns and $s$ rows, the decreasing subsequence has length at most $r+s$. 
    
    The next step is followed directly from Erd\"{o}s-Szekeres theorem (\cref{thm:erdos-szekeres}). $\sigma$ has exactly $r\cdot s/4$ elements. Let $t = \lfloor \frac{r\cdot s}{8(r+s)} \rfloor$ and $l = r+s+1$. Then 
    
    \[t\cdot l -t-l +1  = \lfloor \frac{r\cdot s}{8(r+s)} \rfloor\cdot (r+s+1) -t-r-s \le \frac{r\cdot s}{4}. \]
    
    Since there are no decreasing subsequence in $\sigma$ with length $l = r+s+1$, by Erd\"{o}s-Szekeres theorem, there must exists an increasing subsequence with length $t$. 
    
    Given an increasing subsequence of length $t$, consider the positions of any two consecutive elements in this increasing subsequence, say $(i_1, j_1)$ and $(i_2, j_2)$. By the construction of $\sigma$, we know $j_1 \le j_2$. If $j_1 = j_2$, then by $f(i_1, j_2) < f(i_2, j_2)$,  we must have $i_1 < i_2$. But this is impossible since elements in the same column are arranged in decreasing order. Thus, the only possible case is $(i_1, j_1) < (i_2, j_2)$. The increasing subsequence in $\sigma$ corresponding to a sequence of increasing positions $(i_1, j_1) < (i_2, j_2)< \cdots < (i_t, j_t)$.

\end{proof}

We restate the main result here.

\TypeTwoOrderThm*

In the following, we sometimes view sequences as vectors. For example, a sequence $\{0,1\}^n$ can also be viewed as a vector of dimension $n$.

\begin{proof}
	We first assume $r\cdot s = n$. Our argument can be extended to the case where $r\cdot s < n$. This is because our lower bound is only related to $r$ and $s$.Given a lower bound for sequence of length $r\cdot s$,  we can put dummy elements in the last $ n- r\cdot s$ positions and our lower bound also holds for the sequence of length $n$. 
 
    Without loss of generality, we assume $s$ can be divided by $36$ and $r$ is an even number. This is because our lower bound is asymptotic, if $s$ can not be divided by $36$, we can consider the largest multiple of 36 smaller than s and the similar for $r$. 
	
	In the following, let $p = s/36$ and $q = r/2$. 
	
	Our proof needs to use two asymptotically good error-correcting codes, $\Cone$ and $\Ctwo$. Here, $\Cone\subseteq \{0,1\}^p$ is over binary alphabet. It has codeword length $p$ and  distance $p/4$.  $\Cone$ can be viewed as a subset of $\{0,1\}^p$. $\Ctwo\subseteq (\Cone)^q$ uses $\Cone$ as its alphabet set. It has codeword length $q$ and distance $q/2$. Since both $\Cone$ and $\Ctwo$ are asymptotically good, $|\Cone| = 2^{\Omega(s)}$ and $|\Ctwo| = 2^{\Omega(r\cdot s)}$.
	
	Assume there are two parties, Alice and Bob, each holds a codeword of $\Ctwo$, say Alice holds $u$ and Bob holds $v$. We now show how to use $\Ctwo$ as a fooling set to obtain our lower bound. The high level idea is, we build a sequence $\sigma(u,v) \in\Z^n$  depending on $u$ and $v$, such that $\sigma$ can be divided evenly into $r$ blocks each of length $s$. Those odd blocks are only depending on $u$ (thus only known to Alice) and those even blocks are only depending on $v$ (thus only known to Bob).
 
	
	Our construction guarantees that, if $u=v$, then $\lis(\sigma) = s/9 + 3r + \min(p,q)$, and if $u \neq v$, $\lis(\sigma) = s/9 + 3r + \min(p,q) + \frac{pq}{8(p+q)}$. We first give the construction of $\sigma$.
	

	\paragraph{Construction of $\sigma$.} We first build a 01 matrix $M\in \{0,1\}^{\frac{s}{4} \times 4r}$ using $u$ and $v$. $M$ has $s/4$ rows and $4r$ columns.  We can group the elements of the matrix into $r$ blocks such that, for $1\le i \le r$, the $i$-th block are elements from the 4 columns with number $4i-3, 4i-2, 4i-1, 4i$. In the matrix $M$, the odd blocks are only determined by $u$ (the codeword known to Alice) and the even blocks are only determined by $v$ (the codeword known to Bob). A pictorial representation is given in \cref{fig:M}.
	
	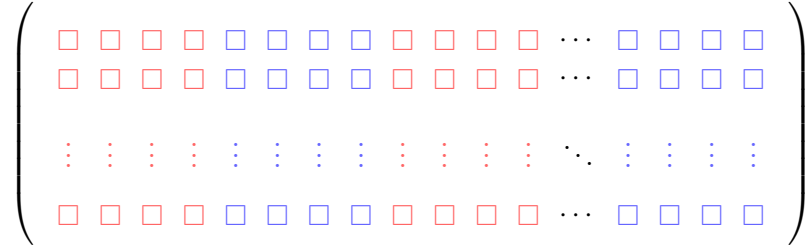
\begin{figure}[h]
	\begin{center}
		\begin{tikzpicture}
			
				\tikzstyle{column 1} = [red!60]
				\tikzstyle{column 2} = [red!60]
				\tikzstyle{column 3} = [red!60]
				\tikzstyle{column 4} = [red!60]
				\tikzstyle{column 5} = [blue!60]
				\tikzstyle{column 6} = [blue!60]
				\tikzstyle{column 7} = [blue!60]
				\tikzstyle{column 8} = [blue!60]
				\tikzstyle{column 9} = [red!60]
				\tikzstyle{column 10} = [red!60]
				\tikzstyle{column 11} = [red!60]
				\tikzstyle{column 12} = [red!60]
				
				\tikzstyle{column 14} = [blue!60]
				\tikzstyle{column 15} = [blue!60]
				\tikzstyle{column 16} = [blue!60]
				\tikzstyle{column 17} = [blue!60]

			     \matrix (M)[ left delimiter = (, right delimiter = )]
				{
					\node{$\square$}; &\node{$\square$}; &\node{$\square$}; &\node{$\square$}; &\node{$\square$}; &\node{$\square$}; &\node{$\square$}; &\node{$\square$}; &\node{$\square$}; &\node{$\square$}; &\node{$\square$}; &\node{$\square$}; &\node{$\cdots$};&\node{$\square$};&\node{$\square$};&\node{$\square$};&\node{$\square$};\\
					\node{$\square$}; &\node{$\square$}; &\node{$\square$}; &\node{$\square$}; &\node{$\square$}; &\node{$\square$}; &\node{$\square$}; &\node{$\square$}; &\node{$\square$}; &\node{$\square$}; &\node{$\square$}; &\node{$\square$}; &\node{$\cdots$};&\node{$\square$};&\node{$\square$};&\node{$\square$};&\node{$\square$};\\
					\node{}; \\
					\node {\vdots}; & \node {\vdots};& \node {\vdots};& \node {\vdots};& \node {\vdots};& \node {\vdots};& \node {\vdots};& \node {\vdots};&\node {\vdots}; & \node {\vdots};& \node {\vdots};& \node {\vdots};& \node {$\ddots$};& \node {\vdots}; & \node {\vdots};& \node {\vdots};& \node {\vdots};\\
					\node{}; \\
					\node{$\square$}; &\node{$\square$}; &\node{$\square$}; &\node{$\square$}; &\node{$\square$}; &\node{$\square$}; &\node{$\square$}; &\node{$\square$}; &\node{$\square$}; &\node{$\square$}; &\node{$\square$}; &\node{$\square$}; &\node{$\cdots$};&\node{$\square$};&\node{$\square$};&\node{$\square$};&\node{$\square$};\\
				};
		\end{tikzpicture}
	\end{center}
	\caption{pictorial representation of $M$. Each small square represents an element of $M$. We partition elements of $M$ into $r$ blocks where each block contains 4 consecutive columns. Odd blocks (elements in red) are determined by Alice's codeword $u$ and even blocks (elements in blue) are determined by Bob's codeword, $v$.  }
 \label{fig:M}
\end{figure} 
	
	We turn this 01 matrix $M$ into a sequence of integers $\sigma(M)$ as following. We first build a new matrix $M'\in \Z^{\frac{s}{4}\times 4r}$, such that, for each position $(i,j)\in [\frac{s}{4}]\times [4r]$, we set 

    \begin{equation}
    \label{eqn:M_value}
        M'_{i,j} = \begin{cases}
            0, & \text{if } M_{i,j} =  0, \\
            4r\cdot (i-1) + j, & \text{if } M_{i,j} = 1.
        \end{cases}
    \end{equation}
 
    Then, $\sigma$ is the concatenation of elements in $M'$ column by column. Specifically, let $\sigma^{(j)} $ be the concatenation of all symbols in the $j$-th column of $M'$, i.e. $\sigma^{(j)} = M'_{1, j}\circ M'_{2, j}\circ \cdots \circ M'_{s/4, j}$ . The final sequence $\sigma$ we get is
 
	$$\sigma(M) = \sigma^{(1)}\circ \cdots \circ \sigma^{(r)}.$$

	\paragraph{View matrix $M$ as a grid graph.} We can view the matrix $M$ as a \textbf{directed} grid graph, such that for each position $(i,j)$, there are two out going edges from this it connecting to $(i+1,j)$ and $(i,j+1)$ respectively (assuming they exist). A \textbf{path} in the matrix $M$ is a path in the corresponding graph that goes from $(1,1)$ to $(\frac{s}{4}, 4r)$ (from left-top to right bottom). We define the \textbf{weight} of a path is the number of $1$ node (position with value $1$) the path covers. We say a node is \textbf{covered} by a path if the path goes through that node.

	\begin{claim}
    \label{clm:weight_lis_eq}
	$\lis(\sigma(M))$ is equal to the largest weight of any path in $M$. 
	\end{claim}
	
	\begin{proof}
     First, for any path in the grid graph, non-zero nodes it covers are in increasing order. This is because any path can only go down or right. By our assignment of the values in $M'$ (\cref{eqn:M_value}), non-zero elements in any row or column are in strict increasing order. Also, since $\sigma$ is the concatenation of all columns of $M'$ from left to right. For any path, nodes it covers is a subsequence of $\sigma$. Thus, we can say for any path with weight $w$ in the grid graph, there is an increasing subsequence of $\sigma$ with length $w$. 
     
	 On the other hand, we show that for any increasing subsequence, there is a path covering all the corresponding nodes in the grid graph. To see this, we show that for any two consecutive nodes $(i,j)$, $(i', j')$ in the increasing subsequence, we have $i\le i' $ and $j \le j'$. $j\le j' $ is by definition since we concatenate columns from left to right. Since it is an increasing subseqeunce, we must have $4r\cdot (i-1) + j < 4r\cdot (i'-1) + j'$. By the fact that there are only $4r$ columns, we have $i\le i'$. Thus, $(i', j')$ must appear on the down-right side of $(i,j)$. We can find a path in the grid graph connecting them. This shows for any increasing subsequence of $\sigma$ with length $w$, there is a path in the grid graph with weight $w$. 
  
	\end{proof}
	
	\paragraph{Constructing $M$ with $u$ and $v$.} In the following, we show how to build the matrix $M$ with $u,v$ (codewords hold by Alice and Bob) and then prove Claim~\ref{clm:lis_upper_bound} and Claim~\ref{clm:lis_lower_bound}, which yields \cref{thm:lis_type2}.
	
	We start with Alice's part. Alice holds $u\in \Ctwo$, or it can be viewed as $q$ codewords of $\Cone$. We denote them by $u^{(1)}, u^{(2)}, \dots, u^{(q)}$. Alice has control over all $q = r/2$ odd blocks. In our construction, the $(2i-1)$-th block (or the $i$-th block hold by Alice) is determined by $u^{(i)}$. The construction is the same for all $i$. 

	We take $2i-1$-th block for an example. There are $4$ columns in this block, i.e. columns with number $8i-7$, $8i-6$, $8i-5$ and $8i-4$. Given $u^{(i)}\in \Cone \subseteq \{0,1\}^{p}$, we turn it into $\bar{u}^{(i)}\in \{0,1\}^{9p}$  such that each $1$ in $u^{(i)}$ is replaced by $(1,1,0,0,0,0,1,1,0)$ and each $0$ is replaced by $(0,0,1,1,1,1,0,0,0)$. This is the $(8i-4)$-th column of $M$. For the other 3 columns, they all have $1$ at rows $9j$ for $j \ge 1$ and $0$ everywhere else.  
	 
	For Bob, the construction is similar. We take $2i$-th block for an example. Given each $v^{(i)}\in \Cone \subseteq \{0,1\}^{p}$, we turn it into $\bar{v}^{(i)}\subseteq \{0,1\}^{9p}$ such that each $1$ in $v^{(i)}$ is replaced by $(1,1,0,0,0,0,1,1,0)$ and each $0$ is replaced by $(0,0,1,1,1,1,0,0,0)$, which is the same as Alice's transform. Then we use $\bar{v}^{(i)}$ as the first column of the $2i$-th block (Note: for Alice's construction, we use it as the last column of the corresponding block). For the other 3 columns, they all have $1$ at rows $9j$ for $j\ge 1$ and $0$ everywhere else.  
	
	Let us see an example. Say $p = q = 2$,  $u^{(1)} = (0,1)$, $u^{(2)} = (1,1)$, $v^{(1)} = (0,0)$ and $v^{(2)} = (1,0)$, then the matrix $M$ is given in \cref{fig:M_example}. Notice that $M$ is a matrix with $s/4 = 9p = 18$ rows and $4r = 8q = 16$ columns. In \cref{fig:M_example}, only elements in red are determined by $u$ and $v$. The rest of the matrix is fixed for any $u$ and $v$. In other words, $u$ and $v$ determined the content of $p \cdot q$ sub-matrices in $M$, each of size $8\times 2$.

    \begin{figure}[ht]
	\begin{center}
		\begin{tikzpicture}
			\scriptsize
			\matrix (m)[
			matrix of math nodes,
			nodes in empty cells,
			left delimiter=(,
			right delimiter=),
			] {
				0&&  0&&  0&&  \color{red}0&&  \color{red}0&&  0&&  0&&  0&&  0&&  0&&  0&&  \color{red}1&&  \color{red}1&&  0&&  0 &&0  \\
	0&&  0&&  0&&  \color{red}0&&  \color{red}0&&  0&&  0&&  0&&  0&&  0&&  0&&  \color{red}1&&  \color{red}1&&  0&&  0 &&0  \\
	0&&  0&&  0&&  \color{red}1&&  \color{red}1&&  0&&  0&&  0&&  0&&  0&&  0&&  \color{red}0&&  \color{red}0&&  0&&  0 &&0  \\
	0&&  0&&  0&&  \color{red}1&&  \color{red}1&&  0&&  0&&  0&&  0&&  0&&  0&&  \color{red}0&&  \color{red}0&&  0&&  0 &&0  \\
	0&&  0&&  0&&  \color{red}1&&  \color{red}1&&  0&&  0&&  0&&  0&&  0&&  0&&  \color{red}0&&  \color{red}0&&  0&&  0 &&0  \\
	0&&  0&&  0&&  \color{red}1&&  \color{red}1&&  0&&  0&&  0&&  0&&  0&&  0&&  \color{red}0&&  \color{red}0&&  0&&  0 &&0  \\
	0&&  0&&  0&&  \color{red}0&&  \color{red}0&&  0&&  0&&  0&&  0&&  0&&  0&&  \color{red}1&&  \color{red}1&&  0&&  0 &&0  \\
	0&&  0&&  0&&  \color{red}0&&  \color{red}0&&  0&&  0&&  0&&  0&&  0&&  0&&  \color{red}1&&  \color{red}1&&  0&&  0 &&0  \\
	\color{blue}1&& \color{blue}1&& \color{blue}1&& \color{blue}0&& \color{blue}0&& \color{blue}1&& \color{blue}1&& \color{blue}1&&  \color{blue}1&& \color{blue}1&& \color{blue}1&& \color{blue}0&& \color{blue}0&& \color{blue}1&& \color{blue}1&& \color{blue}1  \\
	0&&  0&&  0&&  \color{red}1&&  \color{red}0&&  0&&  0&&  0&&  0&&  0&&  0&&  \color{red}1&&  \color{red}0&&  0&&  0 &&0  \\
	0&&  0&&  0&&  \color{red}1&&  \color{red}0&&  0&&  0&&  0&&  0&&  0&&  0&&  \color{red}1&&  \color{red}0&&  0&&  0 &&0  \\
	0&&  0&&  0&&  \color{red}0&&  \color{red}1&&  0&&  0&&  0&&  0&&  0&&  0&&  \color{red}0&&  \color{red}1&&  0&&  0 &&0  \\
	0&&  0&&  0&&  \color{red}0&&  \color{red}1&&  0&&  0&&  0&&  0&&  0&&  0&&  \color{red}0&&  \color{red}1&&  0&&  0 &&0  \\
	0&&  0&&  0&&  \color{red}0&&  \color{red}1&&  0&&  0&&  0&&  0&&  0&&  0&&  \color{red}0&&  \color{red}1&&  0&&  0 &&0  \\
	0&&  0&&  0&&  \color{red}0&&  \color{red}1&&  0&&  0&&  0&&  0&&  0&&  0&&  \color{red}0&&  \color{red}1&&  0&&  0 &&0  \\
	0&&  0&&  0&&  \color{red}1&&  \color{red}0&&  0&&  0&&  0&&  0&&  0&&  0&&  \color{red}1&&  \color{red}0&&  0&&  0 &&0  \\
	0&&  0&&  0&&  \color{red}1&&  \color{red}0&&  0&&  0&&  0&&  0&&  0&&  0&&  \color{red}1&&  \color{red}0&&  0&&  0 &&0  \\
	\color{blue}1&& \color{blue}1&& \color{blue}1&& \color{blue}0&& \color{blue}0&& \color{blue}1&& \color{blue}1&& \color{blue}1&&  \color{blue}1&& \color{blue}1&& \color{blue}1&& \color{blue}0&& \color{blue}0&& \color{blue}1&& \color{blue}1&& \color{blue}1 \\  
			} ;
			
		\draw[thick, green, opacity = 1, -](m-1-11.north west)-- (m-8-11.south west)-- (m-8-21.south east) -- (m-1-21.north east) -- (m-1-11.north west);
		\draw[thick, green, opacity = 1, -](m-10-11.north west)-- (m-17-11.south west)-- (m-17-21.south east) -- (m-10-21.north east) -- (m-10-11.north west);
		\draw[thick, green, opacity = 1, -](m-1-1.north west)-- (m-8-1.south west)-- (m-8-5.south east) -- (m-1-5.north east) -- (m-1-1.north west);
		\draw[thick, green, opacity = 1, -](m-10-1.north west)-- (m-17-1.south west)-- (m-17-5.south east) -- (m-10-5.north east) -- (m-10-1.north west);
		\draw[thick, green, opacity = 1, -](m-1-27.north west)-- (m-8-27.south west)-- (m-8-31.south east) -- (m-1-31.north east) -- (m-1-27.north west);
		\draw[thick, green, opacity = 1, -](m-10-27.north west)-- (m-17-27.south west)-- (m-17-31.south east) -- (m-10-31.north east) -- (m-10-27.north west);

		\draw[thick, red, opacity = 1, -](m-1-7.north west)-- (m-8-7.south west)-- (m-8-9.south east) -- (m-1-9.north east) -- (m-1-7.north west);
		\draw[thick, red, opacity = 1, -](m-10-7.north west)-- (m-17-7.south west)-- (m-17-9.south east) -- (m-10-9.north east) -- (m-10-7.north west);
		\draw[thick, red, opacity = 1, -](m-1-23.north west)-- (m-8-23.south west)-- (m-8-25.south east) -- (m-1-25.north east) -- (m-1-23.north west);
		\draw[thick, red, opacity = 1, -](m-10-23.north west)-- (m-17-23.south west)-- (m-17-25.south east) -- (m-10-25.north east) -- (m-10-23.north west);
		
		\end{tikzpicture}
	\end{center}
	\caption{matrix $M$ when $u^{(1)} = (0,1)$, $u^{(2)} = (1,1)$, $v^{(1)} = (0,0)$ and $v^{(2)} = (1,0)$.}
 \label{fig:M_example}
\end{figure}
	
	We look at the sub-matrices determined by $u$ and $v$. For each pair of indices $(i, j)\in [p] \times [q]$, the sub-matrix determined by the value of $u^{(j)}_i$ and $v^{(j)}_i$ has 4 cases: 
	\begin{enumerate}
	    \item $u^{(j)}_i = 1$ and $v^{(j)}_i = 0$.
	    \item $u^{(j)}_i = v^{(j)}_i = 1$.
	    \item $u^{(j)}_i = v^{(j)}_i = 0$.
	    \item $u^{(j)}_i = 0$ and $v^{(j)}_i = 1$.
	\end{enumerate}

    A pictorial depiction of these 4 cases (from left to right) is given in \cref{fig:M_detail}. As shown in \cref{fig:M_detail}, if $u^{(j)}_i = v^{(j)}_i$, any path in the sub-matrix has weight at most $5$ (marked in red). However, if $u^{(j)}_i \neq v^{(j)}_i$, there exists a path in the sub-matrix with weight 6 (marked in green). 

    \begin{figure}[ht]
	\begin{center}
		\begin{tikzpicture}
			\scriptsize
			\matrix (m)[
			matrix of math nodes,
			nodes in empty cells,
			left delimiter=(,
			right delimiter=)
			] {
				1 && 0  \\
				1 && 0  \\
				0 && 1  \\
				0 && 1  \\
				0 && 1  \\
				0 && 1  \\
				1 && 0  \\
				1 && 0  \\
			} ;
		
			\draw[thick, green, opacity = 0.5, ->](m-1-1.center)-- (m-2-1.center)-- (m-2-3.center) -- (m-8-3.center);

			\matrix (m1)[
			matrix of math nodes,
			nodes in empty cells,
			left delimiter=(,
			right delimiter=),
			right=of m
			] {
				1 && 1  \\
				1 && 1  \\
				0 && 0  \\
				0 && 0  \\
				0 && 0  \\
				0 && 0  \\
				1 && 1  \\
				1 && 1  \\
			} ;
		
			\draw[thick, red, opacity = 0.5, ->](m1-1-1.center)-- (m1-2-1.center)-- (m1-2-3.center) -- (m1-8-3.center);

			\matrix (m2)[
			matrix of math nodes,
			nodes in empty cells,
			left delimiter=(,
			right delimiter=),
			right=of m1
			] {
				0 && 0  \\
				0 && 0  \\
				1 && 1  \\
				1 && 1  \\
				1 && 1  \\
				1 && 1  \\
				0 && 0  \\
				0 && 0  \\
			} ;
	
			\draw[thick, red, opacity = 0.5, ->](m2-1-1.center)-- (m2-3-1.center)-- (m2-3-3.center) -- (m2-8-3.center);
			
			\matrix (m3)[
			matrix of math nodes,
			nodes in empty cells,
			left delimiter=(,
			right delimiter=),
			right=of m2
			] {
				0 && 1  \\
				0 && 1  \\
				1 && 0  \\
				1 && 0  \\
				1 && 0  \\
				1 && 0  \\
				0 && 1  \\
				0 && 1  \\
			} ;

			\draw[thick, green, opacity = 0.5, ->](m3-1-1.center)-- (m3-6-1.center)-- (m3-6-3.center) -- (m3-8-3.center);
		
		\end{tikzpicture}
	\end{center}
	\caption{4 cases of the submatrix determined by $u^{(j)}_i$ and $v^{(j)}_i$.}
 \label{fig:M_detail}
\end{figure}
	
	To simplify the argument a bit, we call the red sub-matrices in \cref{fig:M_example} \textit{key sub-matrices} since everything else are invariant and only these sub-matrices are determined by $u$ and $v$. There are $p\cdot q$ key sub-matrices. We label them by $(i,j)$ such that the key sub-matrix $(i,j)$ is determined by $u^{(j)}_i $ and $v^{(j)}_i$. Here, $1\le i \le p$ and $1\le j \le q$.
	
     
    A path in $M$ can go through those key sub-matrices to gain higher weight. It can also cover the 1's outside these key sub-matrices. An abstraction of $M$ is given in \cref{fig:M_abstraction}. In \cref{fig:M_abstraction}, key sub-matrices are represented as red rectangles and the blue arrow are the 1's that paths can take. They are all on rows with row-number that is a multiple of 9. 0's are not shown in the figure since paths going through these nodes won't gain any weight. We are only interested in paths with largest weight.
    
    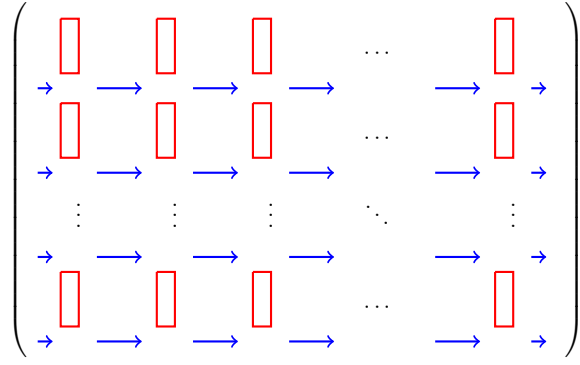
\begin{figure}[ht]
	\begin{center}
		\begin{tikzpicture}
			\scriptsize
			\matrix (m)[
			matrix of math nodes,
			nodes in empty cells,
			left delimiter=(,
			right delimiter=),
			] {
	&&  &&  &&  &&  &&  &&  &&  &&  &&  &&  &&  &&  &&  &&  &&   \\
	&&  &&  &&  &&  &&  &&  &&  &&  &&  &&  &&  &&  &&  &&  &&  \\
	&&  &&  &&  &&  &&  &&  &&  &&  &&  && \cdots &&  &&  &&  &&  &&  \\
	&&  &&  &&  &&  &&  &&  &&  &&  &&  &&  &&  &&  &&  &&  &&  \\
	&&  &&  &&  &&  &&  &&  &&  &&  &&  &&  &&  &&  &&  &&  &&  \\
	&&  &&  &&  &&  &&  &&  &&  &&  &&  &&  &&  &&  &&  &&  &&  \\
	&&  &&  &&  &&  &&  &&  &&  &&  &&  &&  &&  &&  &&  &&  &&  \\
	&&  &&  &&  &&  &&  &&  &&  &&  &&  && \cdots &&  &&  &&  &&  &&  \\
	&&  &&  &&  &&  &&  &&  &&  &&  &&  &&  &&  &&  &&  &&  &&  \\
	&&  &&  &&  &&  &&  &&  &&  &&  &&  &&  &&  &&  &&  &&  &&  \\
	&& \vdots &&  &&  && \vdots &&  &&  && \vdots &&  &&  && \ddots &&  &&  &&  &&  \vdots &&  \\
	&&  &&  &&  &&  &&  &&  &&  &&  &&  &&  &&  &&  &&  &&  &&  \\
	&&  &&  &&  &&  &&  &&  &&  &&  &&  &&  &&  &&  &&  &&  &&  \\
	&&  &&  &&  &&  &&  &&  &&  &&  &&  &&  &&  &&  &&  &&  &&  \\
	&&  &&  &&  &&  &&  &&  &&  &&  &&  &&  &&  &&  &&  &&  &&  \\
	&&  &&  &&  &&  &&  &&  &&  &&  &&  && \cdots &&  &&  &&  &&  &&  \\
	&&  &&  &&  &&  &&  &&  &&  &&  &&  &&  &&  &&  &&  &&  &&  \\
	&&  &&  &&  &&  &&  &&  &&  &&  &&  &&  &&  &&  &&  &&  && \\  
			} ;
			
		
		\draw[thick, blue, opacity = 1, ->](m-5-1.west)-- (m-5-2.west);
		
		\draw[thick, blue, opacity = 1, ->](m-10-1.west)-- (m-10-2.west);
		
		\draw[thick, blue, opacity = 1, ->](m-13-1.west)-- (m-13-2.west);
		
		\draw[thick, blue, opacity = 1, ->](m-18-1.west)-- (m-18-2.west);

		\draw[thick, red, opacity = 1, -](m-1-2.center)-- (m-4-2.center)-- (m-4-3.center) -- (m-1-3.center) -- (m-1-2.center);
		
		\draw[thick, red, opacity = 1, -](m-6-2.center)-- (m-9-2.center)-- (m-9-3.center) -- (m-6-3.center) -- (m-6-2.center);
		
		\draw[thick, red, opacity = 1, -](m-14-2.center)-- (m-17-2.center)-- (m-17-3.center) -- (m-14-3.center) -- (m-14-2.center);

		
		\draw[thick, blue, opacity = 1, ->](m-5-4.center)-- (m-5-7.center);
		
		\draw[thick, blue, opacity = 1, ->](m-10-4.center)-- (m-10-7.center);
		
		\draw[thick, blue, opacity = 1, ->](m-13-4.center)-- (m-13-7.center);
		
		\draw[thick, blue, opacity = 1, ->](m-18-4.center)-- (m-18-7.center);

		\draw[thick, red, opacity = 1, -](m-1-8.center)-- (m-4-8.center)-- (m-4-9.center) -- (m-1-9.center) -- (m-1-8.center);
		
		\draw[thick, red, opacity = 1, -](m-6-8.center)-- (m-9-8.center)-- (m-9-9.center) -- (m-6-9.center) -- (m-6-8.center);
		
		\draw[thick, red, opacity = 1, -](m-14-8.center)-- (m-17-8.center)-- (m-17-9.center) -- (m-14-9.center) -- (m-14-8.center);
		
	    
	    \draw[thick, blue, opacity = 1, ->](m-5-10.center)-- (m-5-13.center);
		
		\draw[thick, blue, opacity = 1, ->](m-10-10.center)-- (m-10-13.center);
		
		\draw[thick, blue, opacity = 1, ->](m-13-10.center)-- (m-13-13.center);
		
		\draw[thick, blue, opacity = 1, ->](m-18-10.center)-- (m-18-13.center);
		
		\draw[thick, red, opacity = 1, -](m-1-14.center)-- (m-4-14.center)-- (m-4-15.center) -- (m-1-15.center) -- (m-1-14.center);
		
		\draw[thick, red, opacity = 1, -](m-6-14.center)-- (m-9-14.center)-- (m-9-15.center) -- (m-6-15.center) -- (m-6-14.center);
		
		\draw[thick, red, opacity = 1, -](m-14-14.center)-- (m-17-14.center)-- (m-17-15.center) -- (m-14-15.center) -- (m-14-14.center);
		

		 \draw[thick, blue, opacity = 1, ->](m-5-16.center)-- (m-5-19.center);
		
		\draw[thick, blue, opacity = 1, ->](m-10-16.center)-- (m-10-19.center);
		
		\draw[thick, blue, opacity = 1, ->](m-13-16.center)-- (m-13-19.center);
		
		\draw[thick, blue, opacity = 1, ->](m-18-16.center)-- (m-18-19.center);
		
		
		\draw[thick, blue, opacity = 1, ->](m-5-24.center)-- (m-5-27.center);
		
		\draw[thick, blue, opacity = 1, ->](m-10-24.center)-- (m-10-27.center);
		
		\draw[thick, blue, opacity = 1, ->](m-13-24.center)-- (m-13-27.center);
		
		\draw[thick, blue, opacity = 1, ->](m-18-24.center)-- (m-18-27.center);
		
		\draw[thick, red, opacity = 1, -](m-1-28.center)-- (m-4-28.center)-- (m-4-29.center) -- (m-1-29.center) -- (m-1-28.center);
		
		\draw[thick, red, opacity = 1, -](m-6-28.center)-- (m-9-28.center)-- (m-9-29.center) -- (m-6-29.center) -- (m-6-28.center);
		
		\draw[thick, red, opacity = 1, -](m-14-28.center)-- (m-17-28.center)-- (m-17-29.center) -- (m-14-29.center) -- (m-14-28.center);
		
		\draw[thick, blue, opacity = 1, ->](m-5-30.center)-- (m-5-31.center);
		
		\draw[thick, blue, opacity = 1, ->](m-10-30.center)-- (m-10-31.center);
		
		\draw[thick, blue, opacity = 1, ->](m-13-30.center)-- (m-13-31.center);
		
		\draw[thick, blue, opacity = 1, ->](m-18-30.center)-- (m-18-31.center);

		\end{tikzpicture}
	\end{center}
	\caption{abstraction of $M$}
 \label{fig:M_abstraction}
\end{figure}
     
    We show that the largest weight of any path in $M$ can vary depending on whether $u=v$. We have the following two claims. 
	
	\begin{claim}
    \label{clm:lis_upper_bound}
	 If $u=v$, then $\lis(\sigma(M)) \le 4p + 6q + \min(p,q) $. 
	\end{claim}
	
	\begin{proof}
	 By Claim~\ref{clm:weight_lis_eq}, only need to show that any path in $M$ has weight at most $4p + 6q + \min(p,q)$. 
	 
     Let $\path$ be one of the path in $M$ that has largest weight. Since $u=v$, for each key sub-matrix, if $\path$ goes from top-left to bottom-right of this key sub-matrix, it will gain a weight of $5$. If $\path$ only goes through one column of the key sub-matrix, it will gain a total weight of $4$. 
     
     We have the following simple observation. We can assume $\path$ does not go beyond key sub-matrices and blue arrows (except for the first few steps). This is because every time $\path$ does this, it can only increase the total weight by at most 1, but it will waste the weight it could gain by going through the blue arrow (which is 6) or the key sub-matrix (which is at least 4). 
     
	 If $p \ge q$, the best strategy is to gain a weight of 5 in as much key sub-matrices as possible. But we can only do this for $q$ times since every time it gain a total weight of 5 in one key sub-matrix, it has to go to a new column. The total weight it can gain in blue arrows is at most $6q$ since we assume it will not go to the zero region. We can gain an additional weight of $4(p-q)$ in key sub-matrices (in these sub-matrices, $\path$ only uses one column). Thus, the total weight in this case is at most $5q + 6q + 4(p-q) = 4p + 6q + q$. 
	 
	 If $p < q$, the best strategy is again to gain a weight of 5 in as much key sub-matrices as possible. We can do this for $p$ times. The total weight it can gain in blue arrows is at most $6q$. Thus, the total weight in this is at most $5p + 6q$.
	 
	 Combine these two cases, we conclude the largest weight of any path in $M$ is at most $4p + 6q + \min(p,q) $.
	 
	\end{proof}

	\begin{claim}
    \label{clm:lis_lower_bound}
	If $u\neq v$, then $\lis(\sigma(M)) \ge 4p + 6q -3 + \min(p,q) + \frac{pq}{16(p+q)}$.
	\end{claim}
	
	\begin{proof}
	 By Claim~\ref{clm:weight_lis_eq}, only need to show there exists a path in $M$ with weight at least $4p + 6q -3 + \min(p,q) + \frac{pq}{16(p+q)}$.
	 
	 Since $u$ and $v$ are chosen from $\Ctwo$, which is an error-correcting code with distance $q/2$. If $u\ne v$, there are $q/2$ positions what $u$ and $v$ that are not equal. Let $j$ be on of them ($u^{(j)}\ne v^{(j)}$). Since $u^{(j)}$ and $ v^{(j)}$ are both codewords of $\Cone\subseteq \{0,1\}^p$. There are $p/4$ positions that $u^{(j)}$ and $ v^{(j)}$ does not equal. 
	 
	 Let $A$ be a $p\times q$ matrix such that $A_{i,j} = 1$ if  $u^{(j)} \ne v^{(j)}$ and  $A_{i,j} = 0$ if  $u^{(j)} = v^{(j)}$. Then in matrix $A$ there are at least $q/2$ columns with at least $p/4$ 1's. Then, if we only look at those columns with at least $p/4$ 1's, by \cref{lem:erdos-szekeres}, there are  $t = \lfloor \frac{r\cdot s}{16(r+s)} \rfloor$ positions in the matrix (denoted $(i_1, j_1), (i_2, j_2), \dots, (i_t, j_t)$) with value 1, such that  $(i_1, j_1) < (i_2, j_2)< \cdots < (i_t, j_t)$.
	 
	 The best strategy for a path to gain largest weight is to go through key sub-matrices that can provide a total weight of 6 as much as possible ($u^{(j)} \ne v^{(j)}$). \cref{lem:erdos-szekeres} guarantees there are at least $t = \lfloor \frac{r\cdot s}{16(r+s)} \rfloor$ that a path can go through. For the remaining, it can go from top-left to bottom-right of $\min{p,q} - t$ key sub-matrices. Each will give a weight of at least 5. If $p>q$, it can gain additional $4(p-q)$ by going through $p-q$ key sub-matrices but using only one column. Along the blue arrows, a path can again at least $6q - 3$ (minus the 3 it misses on the first row). Thus, we can find a path with weight at least $4p + 6q -3 + \min(p,q) + \frac{pq}{16(p+q)}$.
	 
	\end{proof}
	
	By Claim~\ref{clm:lis_upper_bound} and Claim~\ref{clm:lis_lower_bound}, we can conclude that $\Ctwo$ is a fooling set for $\textbf{type 2}$. Even if Alice and Bob can approximate $\lis$ to within a factor of $1+1/400$. 
	
	When $p,q$ are large enough, we have 
	\begin{align*}
	    (1+1/400) \cdot \big(4p + 6q + \min(p,q)\big) \le & 4p + 6q + \min(p,q) + \frac{11\max(p,q)}{400} \\
	    < & 4p + 6q + \min(p,q) + \frac{22\cdot p\cdot q}{400\cdot(p+q)} \\
	    \le & 4p + 6q + \min(p,q) + \frac{pq}{16(p+q)} - 3. \\
	\end{align*}
	
	In this case, they can still not determine whether $u = v$ or not. Since both $\Cone$ and $\Ctwo$ are asymptotically good, $\Ctwo$ has size $2^{\Omega(rs)}$. To deterministically approximate $\lis$ to within a $1+1/400$ factor (know whether $u = v$ or not), any $R$-round streaming algorithm would need a space of at least $\Omega(rs/R)$.




\end{proof}

\bibliographystyle{alpha}
\bibliography{references}

\end{document}